\documentclass[11pt,a4paper]{amsart}
\usepackage[a4paper,centering]{geometry}
\usepackage{cite}
\usepackage{graphicx}
\usepackage{subfig}
\makeatletter
\newcommand{\gd}{\delta}

\renewcommand{\gg}{\gamma}
\newcommand{\gk}{\kappa}
\newcommand{\go}{\omega}
\newcommand{\gs}{\sigma}
\newcommand{\vp}{\varphi}
\newcommand{\gD}{\Delta}
\newcommand{\gG}{\Gamma}
\newcommand{\gL}{\Lambda}
\newcommand{\cB}{\mathcal{B}}
\newcommand{\cV}{\mathcal{V}}
\newcommand{\R}{\mathbb{R}}
\newcommand{\p}{\partial}

\newtheorem{theorem}{Theorem}[section]
\newtheorem{lemma}[theorem]{Lemma}
\newtheorem{corollary}[theorem]{Corollary}

\theoremstyle{remark}
\newtheorem{remark}[theorem]{Remark}

\newcommand{\diag}[1]{\text{diag}(#1)}
\newcommand{\cs}{\ensuremath {\cosh(s)}}
\newcommand{\sis}{\ensuremath {\sinh(s)}}

\newcommand{\sit}{\ensuremath {\sin(\theta)}}
\let\lg\langle
\let\rg\rangle
\let\mb\mathbf

\date{March 21, 2017}
\title{Construction of a Relativistic Ornstein--Uhlenbeck Process}

\author[J.~Potthoff]{J\"urgen Potthoff}
\address{J\"urgen Potthoff\newline
Institut f\"ur Mathematik\newline 
Universit\"at Mannheim\newline
D--68131 Mann\-heim, Germany}
\email{potthoff@uni-mannheim.de}

\author[R.~Schrader]{Robert Schrader}

\copyrightinfo{2017}{J.~Potthoff, R.~Schrader}

\begin{document}
\begin{abstract}
Based on a version of Dudley's Wiener process~\cite{Du65} on the mass 
shell in the momentum Minkowski space of a massive point particle, a model
of a relativistic Ornstein--Uhlenbeck process is constructed by addition
of a specific drift term. The invariant distribution of this momentum process 
as well as other associated processes are computed. 
\end{abstract}

\maketitle
\thispagestyle{empty}
\section{Introduction} \label{sect_Intro}
In 1930 Ornstein and Uhlenbeck~\cite{UhOr30} introduced the stochastic process which 
afterwards carried their name in order to treat the case where the particle undergoing 
a motion of Brownian type has a surrounding medium which is a rarefied gas instead 
of a liquid. They argued that in this case one has to take into account the friction 
that the particle experiences by hitting the gas molecules, which they called 
\emph{Doppler friction}. As a consequence, they proposed an equation of Langevin type 
for the velocity of the particle instead of for its position. They proved that the 
velocity of the so defined motion admits a stationary state which is described by a 
centered normal density. For a discussion of the Ornstein--Uhlenbeck process from the 
physics point of view, especially in comparison to the Einstein--Smoluchowski theory
of Brownian motion, we refer the interested reader to~\cite{Ne67}. For example 
in~\cite{IkWa89} one can find a treatment within the context of It\^o's theory of 
stochastic differential equations.

The ground breaking paper~\cite{Du65} by Dudley in 1965 was the first in which
a relativistic Wiener process has been constructed. Since then a large amount of 
literature on diffusion processes in the frameworks of special and general relativity
has been published. We refer the interested reader especially to the overview 
papers~\cite{DuHa09, Fr14}, to the literature quoted there, and also to 
\cite{Ba10, DuHa05, DuHa05a, Fr09, FrLe06, FrLe11, FrLe12, Ha09, Ha10, He09, He10}.

A construction of a relativistic Ornstein--Uhlenbeck process is provided by~\cite{DeMa97},
based on a relativistic formulation of the Langevin equation. In~\cite{Ha09, Ha10}
Haba has studied a variety of relativistic diffusion processes, and among them also 
processes of Ornstein--Uhlenbeck type. The relation of these papers to the model we 
construct in the present manuscript has still to be worked out. 

We consider the momentum of the particle as the basic dynamical quantity.
Therefore we consider the relativistic four momentum 
\begin{equation*}
	\mb{p}=(p_0,p_1,p_2,p_3)\in\R^4
\end{equation*} 
of a massive point particle, and the special theory of relativity demands that $\mb{p}$ has to
be a point of the mass shell, that is the condition
\begin{equation*}
	p_0^2 - \sum_{i=1}^3 p_i^2 = m^2 c^2,\qquad p_0\ge 0,
\end{equation*}
has always to be fulfilled, where $m$ is the mass of the particle, and $c$ is the speed 
of light in vacuum. It turns out that the mass shell is a $3$ dimensional Riemannian 
manifold, and which therefore is equipped with a canonical, positive definite 
Laplace--Beltrami operator. Therefore we first construct a stochastic process of Wiener
type on the mass shell, which has this Laplace--Beltrami operator as its generator,
and the resulting process is a version of Dudley's Wiener process~\cite{Du65}.
Next, via an It\^o stochastic differential equation, we add a drift of a specific form
(cf.\ section~\ref{sect_OU}) to this Wiener
process in order to imitate the drift which has been introduced by Ornstein and Uhlenbeck
to model the Doppler friction. As consequence, we obtain a stochastic process which
we call \emph{relativistic Ornstein--Uhlenbeck momentum process}, which moves on the
mass shell and admits a stationary state. The \emph{relativistic Ornstein--Uhlenbeck
velocity process} is then defined as prescribed by special relativity, namely as
the space components of the momentum process divided by the energy process (times $c^2$).

The plan of the article is as follows. In section~\ref{sect_LB} we discuss the mass shell 
as a Riemannian manifold and calculate the associated Laplace--Beltrami operator. The 
Wiener process on the mass shell is constructed in section~\ref{sect_BM}, while the 
relativistic Ornstein--Uhlenbeck process is treated in section~\ref{sect_OU}. In 
section~\ref{sect_IM} various invariant measures (or stationary states) are computed 
for the momentum and the velocity process. In appendix~\ref{app_Sim} we describe our 
simulation method for the stochastic differential equations.

\par\vspace{1\baselineskip}\noindent
\textbf{Acknowledgement.}\ JP gratefully acknowledges helpful discussions with 
\emph{Leif D\"o\-ring} and \emph{Andreas Neuenkirch}. The authors are very much 
indebted to \emph{Zbiegniew Haba} for pointing out the work by R.M.~Dudley, 
J.~Franchi, Y.~Le~Jan, and the references \cite{Ha09, Ha10}.

\section{The Mass Shell and its Laplacian}	\label{sect_LB}
Throughout this article we consider a space dimension $d$, which is greater or equal 
to two. For points $\mathbf{p}$ in $\R^{1+d}$ we write their cartesian coordinates as 
$\mb{p} = (p_0,p)$ with $p_0\in\R$, $p=(p_1,\dotsc,p_d)\in\R^d\!$. $\R^{1+d}$ is equipped 
with the Minkowski metric tensor $g_M = \diag{1,-1,\dotsc,-1}$, and inner product
\begin{equation*}
	\lg \mb{p},\mb{q}\rg = (\mb{p}, g_M\cdot \mb{q}) = p_0 q_0 - \sum_{i=1}^d p_i q_i,
\end{equation*}
where ``$\,\cdot\,$'' denotes matrix multiplication, and $(\,\cdot\,,\,\cdot\,)$ 
stands for the euclidean product in~$\R^{1+d}\!$. For $m>0$ define the \emph{mass 
shell}
\begin{equation*}
	\cV_m^d = \bigl\{\mb{p}\in\R^{1+d},\,p_0>0,\,
					\lg \mb{p},\mb{p}\rg = m^2 c^2\bigr\},
\end{equation*}
where $c$ is the speed of light in vacuum. From now on we shall use physical
units such that $c=1$.

Clearly, if $\mb{p}\in\cV_m^d$ then $p_0\ge m$ holds true, and for given $p_0>m$,
$p$ belongs to the $d-1$ dimensional sphere $S^{d-1}_\rho$ of radius
$\rho = (p_0^2-m^2)^{1/2}$. 

It is convenient to coordinatize $\cV_m^d$ by hyperbolic coordinates 
$(s,\go)\in \R_+\times S^{d-1}$, where $S^{d-1}$ denotes the $(d-1)$--dimensional
unit sphere. Namely --- with the exception of the apex $\mb{p} = (m,0,\dotsc,0)$ --- 
every point $\mb{p}=(p_0,p)$ in $\cV_m^d$ can uniquely be written as 
\begin{equation}	\label{hyp_coord}
	p_0 = m\cosh(s),\quad p = m\sinh(s)\,\go,\qquad s>0,\,\go\in S^{d-1}.
\end{equation}
At the apex, i.e., for $s=0$, we simply leave $\go$ undefined. 

To make this more concrete, we let $S^{d-1}$ be parametrized in the usual way by 
angles $\theta_1$, \dots, $\theta_{d-1}$ with $\theta_k\in[0,\pi)$, $k=1$, \dots, 
$d-2$, and $\theta_{d-1}\in [0,2\pi]$. Set $\theta = (\theta_1,\dotsc, \theta_{d-1})$, 
and consider the mapping
\begin{equation*}
	\iota: (s,\theta) \mapsto \mb{p}(s,\theta) 
					= \bigl(m\cs, m\sis\,\go(\theta)\bigr).
\end{equation*}
We may consider this mapping as an immersion of $\cV_m^d$ into $\R^{1+d}$. In 
block form its Jacobian reads
\begin{equation*}
	J = m \begin{pmatrix}
			\sis 			    & 	0		\\
			\cs\,\go(\theta)	& \displaystyle \sis\, 
									\Bigl(\frac{\p \go(\theta)}{\p\theta}\Bigr)
		\end{pmatrix},
\end{equation*}
and $\bigl(\p \go(\theta)/\p \theta\bigr)$ is the Jacobian of the embedding of the
unit sphere $S^{d-1}$, coordinatized by the angles $\theta$, into~$\R^d$. With the
immersion $\iota$ we pull the Minkowski metric $g_M$ back on $\cV_m^d$ yielding
a metric $g_d$, which in matrix form is given by
\begin{equation*}
	g_d = - J^t\cdot g_M\cdot J.
\end{equation*}
The minus sign is chosen for later convenience, and the superscript ``$t$'' stands 
for transposition. Hence
\begin{align*}
	g_d(&s,\theta) = \\
		&= m^2 \begin{pmatrix}
				\cs^2\, \go(\theta)^t\cdot\go(\theta) - \sis^2		
					& \sis\cs\, \go(\theta)^t\cdot 
							\bigl(\frac{\p \go(\theta)}{\p \theta}\bigr)\\[2ex]
						    \sis\cs\, \go(\theta)\cdot 
								\bigl(\frac{\p \go(\theta)}{\p \theta}\bigr)^t
					& \sis^2 \bigl(\frac{\p \go(\theta)}{\p \theta}\bigr)^t
							\cdot\bigl(\frac{\p \go(\theta)}{\p \theta}\bigr)	
		   \end{pmatrix}\\
		&= m^2 \begin{pmatrix}
				1	&	0	\\
				0	&   \sis^2 g_{S^{d-1}}(\theta)
		  \end{pmatrix},
\end{align*}
because $\go(\theta)^t\cdot\go(\theta)=1$, which is also the reason that the
off-diagonal terms vanish. Moreover, $g_{S^{d-1}}(\theta)$ is the usual Riemannian 
metric tensor  of the unit sphere $S^{d-1}$ in $d$ dimensions, written as a matrix 
parametrized by the angle variables $\theta=(\theta_1,\dotsc, \theta_{d-1})$. 

Observe that $g_d$ is positive definite, providing a Riemannian metric 
on~$\cV_m^d$. The associated volume element is
\begin{equation}	\label{RVol}
\begin{split}
	d\text{vol}_d(s,\theta) 
		&= m\,\sinh(s)^{d-1}\,ds\,d\gs_{S^{d-1}}(\theta)\\
		&= m\,\sinh(s)^{d-1} \bigl(\det g_{S^{d-1}}(\theta)\bigr)^{1/2}\,ds\,d\theta,					
\end{split} 
\end{equation}
where $d\gs_{S^{d-1}}$ is the Riemannian surface element of the sphere $S^{d-1}$.
The usual well-known formula (e.g., \cite{He78}, \cite{Jo11}) for the 
Laplace--Beltrami operator $\gD_d$ on $\cV_m^d$ relative to $g_d$ yields
\begin{equation}	\label{LBd}
\begin{split}
	\gD_d &= \frac{1}{m^2 \sis^{d-1}}\,\frac{\p}{\p s}\, 
				\sis^{d-1}\,\frac{\p}{\p s}
					+ \frac{1}{m^2\sis^2}\,\gL_{S^{d-1}}\\
		  &= \frac{1}{m^2}\,\frac{\p^2}{\p s^2} + \frac{d-1}{m^2}\,
				\coth(s)\,\frac{\p}{\p s} + \frac{1}{m^2\sis^2}\,\gL_{S^{d-1}},
\end{split} 
\end{equation}
where $\gL_{S^{d-1}}$ is the standard Laplace--Beltrami operator on the 
sphere~$S^{d-1}$. For $d=2$ we find the explicit form
\begin{equation} \label{LB2}
	\gD_2 = \frac{1}{m^2}\,\frac{\p^2}{\p s^2} + \frac{1}{m^2}\,\coth(s)\,\frac{\p}{\p s} 
				+ \frac{1}{m^2\sis^2}\,\gL_{S^{d-1}}\,\frac{\p^2}{\p \vp^2},
\end{equation}
while for the case $d=3$ of the physical Minkowksi space it reads
\begin{equation}	\label{LB3}
\begin{split}
	\gD_3 = \frac{1}{m^2}\,\frac{\p^2}{\p s^2} &+ \frac{2}{m^2}\,\coth(s)\,\frac{\p}{\p s} \\
		  		&+\frac{1}{m^2\sis^2}\,\Bigl( 
		          \frac{1}{\sit}\, \frac{\p}{\p \theta}\,\sit\,\frac{\p}{\p \theta}
				       + \frac{1}{\sit^2}\, \frac{\p^2}{\p\vp^2}\Bigr).
\end{split}
\end{equation}

\section{Wiener Process on the Mass Shell}	\label{sect_BM}
A natural way to define a Wiener process on a Riemannian manifold is as a stochastic 
process whose generator is one half times the canonical Laplace--Beltrami operator 
on the manifold. The history of Wiener and --- more generally --- diffusion processes 
on manifolds is quite long, and probably the first papers where those by Yosida~\cite{Yo49, Yo52}, 
and by It\^o~\cite{It50}. A turning point in this development was the 
construction by Eells, Elworthy~\cite{EeEl76}, and Malliavin \cite{Ma78} (see also~\cite{El82}), 
based on the rolling map of Cartan. The resulting construction of a diffusion on the 
orthonormal frame bundle is since then one of the basic methods, and it can be found 
in many textbooks such as~\cite{El82, IkWa89, HaTh94, Hs02}.

Dudley~\cite{Du65} has been the first who constructed a Wiener process on $\cV_m^d$. 
His construction uses the theory of convolution semigroups on homogeneous spaces. 
Here we employ a different method which is as follows.
The specific form~\eqref{LBd} of the Laplace--Beltrami operator $\gD_d$, namely 
the fact the differential operator in the $s$--variable involves no dependency on 
the angle variables, and the $s$--dependence of the last term only appears in form 
of a factor, suggests another possible construction via stochastic differential 
equations and a skew product, e.g., \cite[Sect.\ 7.15]{ItMc74}, which we carry
out now. 

In fact, suppose that $\Theta=(\Theta_t,\,t\in\R_+)$ 
is a Wiener process on the sphere $S^{d-1}$ (i.e., with generator 
$1/2\,\gL_{S^{d-1}}$), and $S=(S_t,\,t\in\R_+)$ is a stochastic process on 
$(0,+\infty)$ with continuous paths, and generator $A_d$ defined by
\begin{equation}	\label{GenS}
	A_df(s) =  \frac{1}{2m^2}\,f''(s) 
					+\frac{d-1}{2m^2}\,\coth(s)f'(s),\qquad s\in (0,+\infty),
\end{equation}
for $f\in C^2((0,+\infty))$. The It\^o stochastic differential equation~(SDE) 
associated with the generator $A_d$ is
\begin{equation}	\label{SDE_S}
	dS_t = \frac{d-1}{2m^2}\,\coth(S_t)\,dt + \frac{1}{m}\,dW_t,
\end{equation}
where $W$ is a standard Wiener process on the real line. Below we shall prove 
the existence and uniqueness of solutions of this equation. Consider the 
stochastic time scale
\begin{equation}	\label{tau}
	\tau(t) = \int_0^t \bigl(m\sinh(S_r)\bigr)^{-2}\,dr
\end{equation}
then the skew product
\begin{equation}	\label{skewpr}
	B_t = \bigl(S_t,\Theta_{\tau(t)}\bigr),\qquad t\in\R_+
\end{equation}
of $S$ and $\Theta$ defines a path-continuous stochastic process on $\cV_m^d$
whose generator is $1/2\,\gD_d$, i.e., a Wiener process on $\cV_m^d$. The proof 
is the same as for the spherical Wiener process in Section~7.15 of~\cite{ItMc74}. 
In fact there one can also find a method for the construction of the Wiener process 
$\Theta$ on $S^{d-1}$ by successive applications of skew products. Another 
possibility for the construction of the Wiener process $\Theta$ on $S^{d-1}$ is 
Stroock's method (\cite{St71}, cf.\ also \cite{Hs02}): In this case one takes 
a standard Wiener process in the euclidean space $\R^d$, starts it on the 
embedded sphere $S^{d-1}$, and projects the infinitesimal increments of the 
euclidean Wiener process onto the sphere by the usual orthogonal projections. 
The resulting stochastic differential equation can be solved, and yields another
version of the process $\Theta$.

Therefore, in order to complete our first task, viz.\ to construct a Wiener process
on the mass shell $\cV_m^d$, it remains to construct the process $S$ on $(0,+\infty)$
with generator $A_d$ as in~\eqref{GenS}. In other words, we want to prove that the 
It\^o SDE~\eqref{SDE_S} has unique solutions (in which precise sense will be clarified 
further below). Due to the singularity of the drift term of the SDE~\eqref{SDE_S} 
\begin{equation}	\label{eq_drift_S}
	b_0(s) = \frac{d-1}{2m^2}\, \coth(s),\qquad s>0,
\end{equation}
at $s=0$, one cannot employ the standard theorems on the existence and uniqueness
of SDE's, as they can be found in, e.g., \cite{IkWa89}, \cite{KaSh91}, \cite{ReYo91}.
However, by using results in the book~\cite{ChEn05} by Cherny and Engelbert
we can prove following 

\begin{theorem} \label{thm_ex_un_b0}
For every starting point $s_0\in (0,+\infty)$ the SDE~\eqref{SDE_S} has a pathwise
unique strong solution $S=(S_t,\,t\in\R_+)$ with paths which are $P$--a.s.\ strictly 
positive. Moreover, the solutions are transient in the following sense: For every 
$a>0$ and every initial condition $s_0>a$ the event $\{T_a=+\infty\}$, where $T_a$ 
is the first hitting time of $a$ by $S$, has strictly positive probability, and on 
this set $\lim_{t\to+\infty} S_t=+\infty$ $P$--a.s. 
\end{theorem}

\begin{remark}	\label{rem_ex_un_b0}
We quickly (and somewhat roughly) recall the definition of strong existence and 
pathwise uniqueness --- for an in depth overview of the various notions of existence 
and uniqueness of solutions of stochastic differential equations and their 
interrelations we refer the interested reader to, e.g., \cite[Sect.~1.1]{ChEn05}. 
The existence of a strong solution of~\eqref{SDE_S} means that for any given 
Wiener process $W$ on the real line, there exists a solution $S$ of~\eqref{SDE_S} 
which is adapted to the filtration generated by $W$. Thus the paths of $S$ can 
be considered as adapted functionals of the paths of~$W$. Pathwise uniqueness means 
that if $S$ and $S'$ are two solutions defined on the same probability space with 
the same initial condition, and with the same driving Wiener process $W$, then 
$P(S_t=S'_t,\,t\in\R_+)=1$.
\end{remark}

For the proof of theorem~\ref{thm_ex_un_b0} show first two lemmas. In order to 
simplify our notation for the following discussion we shall temporarily set $m=1$. 
The first step is to prove the analogue statement as in theorem~\ref{thm_ex_un_b0} 
for \emph{weak} existence and uniqueness:

\begin{lemma} \label{lem_wex_wun_b0}
For every starting point $s_0\in (0,+\infty)$ the SDE~\eqref{SDE_S} has a 
weak solution $S=(S_t,\,t\in\R_+)$ which is unique in law, and with paths which 
are $P$--a.s.\ strictly positive. Moreover, the solutions are transient in the following 
sense: For every $a>0$ and every initial condition $s_0>a$ the event $\{T_a=+\infty\}$, 
where $T_a$ is the first hitting time of $a$ by $S$, has strictly positive probability, 
and on this set $\lim_{t\to\infty} S_t=+\infty$ $P$--a.s. 
\end{lemma}

\begin{remark} \label{rem_wex_wun_b0}
Also here we first want to quickly recall the meaning of the existence and uniqueness 
statement. That the stochastic differential equation~\eqref{SDE_S} has a weak 
solution roughly speaking means that on some filtered probability space there exists
a \emph{pair} $(S,W)$ of adapted processes so that the integrated version 
of~\eqref{SDE_S} holds true. Uniqueness in law of the solution means that if 
$(S,W)$ and $(S',W')$ are two such pairs (possibly defined on different probability 
spaces) with the same initial condition, then the laws of $S$ and $S'$ coincide. 
Furthermore we remark in passing that the existence of a weak solution is equivalent 
to the existence of the associated martingale problem (e.g., \cite[Theorem~1.27]{ChEn05} 
or \cite[Proposition~IV.2.1]{IkWa89}).
\end{remark}

\begin{proof}[Proof of lemma~\ref{rem_wex_wun_b0}]
We show that the conditions of theorems~2.16, 4.2, and part~(viii) of theorem~4.6 
in~\cite{ChEn05} hold true. First we remark that for every $a>0$ the drift 
\begin{equation} \label{eq_drift_b0}
	b_0(s) = \frac{d-1}{2}\, \coth(s),\qquad s>0,
\end{equation}
obviously belongs to $L^1_{{\rm loc}}([a,+\infty))$ so that the origin $s=0$ is indeed 
an isolated singularity in the sense of~\cite[Sect.\ 2.1]{ChEn05}. We fix some $a>0$ 
for the remainder of this proof. Next we compute the density $\rho$ of the so-called 
scale functions. Since in our case the diffusion coefficient is equal to~$1$, $\rho$ 
is given by
\begin{equation}	\label{eq_rho_b0}
	\rho(s) = \exp\Bigl(2\int_s^a b_0(u)\,du\Bigr)
			= \frac{\sinh(a)^{d-1}}{\sinh(s)^{d-1}}.
\end{equation}
Furthermore we define the scale functions
\begin{align*}
	\gk_a(s)      &= -\int_s^a \rho(u)\,du,\qquad s\in(0,a],\\
	\gk_\infty(s) &= -\int_s^\infty \rho(u)\,du,\qquad s\in[a,+\infty).
\end{align*}
Since $d\ge 2$, we clearly have from~\eqref{eq_rho_b0} that 
\begin{align}
	\int_0^a \rho(s)\,ds &=+\infty, \label{rho_a}\\
	\int_a^\infty \rho(s)\,ds &< +\infty \label{rho_infty}.
\end{align}
Moreover we claim that the following are true:
\begin{align}
	I_a 		&= \int_0^a \bigl(1+|b_0(s)|\bigr)\, \rho(s)^{-1}\, |\gk_a(s)|\,ds <+\infty,\label{I_a}\\
	I_\infty    &= \int_a^\infty \rho(s)^{-1}\, |\gk_\infty(s)|\,ds = +\infty.\label{I_infty}
\end{align}
The integral~$I_a$ is equal to
\begin{equation*}
	\int_0^a \Bigl(1+ \frac{(d-1)\cosh(s)}{2\sinh(s)}\Bigr) \sinh(s)^{d-1}
				\Bigl(\int_s^a \sinh(u)^{-(d-1)}\,du\Bigr)\,ds.
\end{equation*} 
Since $s\mapsto \sinh(s)$ is convex, we have for all $s\in(0,a]$ the inequalities
\begin{equation}	\label{sinh}
	s \le \sinh(s) \le \sinh(a)\, s.
\end{equation}
Therefore we can estimate as follows
\begin{equation*}
	I_a \le \sinh(a)^{d-1} \int_0^a \Bigl(1 + \frac{d-1}{2}\,\coth(x)\Bigr)\, s^{d-1}
			\Bigl(\int_s^a u^{-(d-1)}\,du\Bigr)\,ds.
\end{equation*}
For $d=2$ we get
\begin{equation*}
	I_a \le \sinh(a)^{d-1} \int_0^a \Bigl(1 + \frac{d-1}{2}\,\coth(s)\Bigr)\, 
										s \ln\Bigl(\frac{a}{s}\Bigr)\,ds<+\infty,
\end{equation*}
while for $d\ge 3$ we find
\begin{equation*}
\begin{split}
	I_a \le \frac{\sinh(a)^{d-1}}{d-2} \int_0^a \Bigl(1 
				+ \frac{d-1}{2}\,&\coth(s)\Bigr)\, s^{d-1}\\[.5ex]
		&\times\,\bigl(s^{-(d-2)} - a^{-(d-2)}\bigr)\,ds <+\infty.
\end{split}
\end{equation*}
Hence~\eqref{I_a} is proved, and together with~\eqref{rho_a} this shows that the
conditions of theorem~2.16 in~\cite{ChEn05} are fulfilled. As a consequence we
get the statement that the SDE~\eqref{SDE_S} has for every starting point $s_0>0$ a 
unique weak solution up to the first hitting time $T_a$, and the solution is a.s.\ strictly 
positive. Moreover --- and this will be more important below --- the singularity at 
$s=0$ is of \emph{type~3} in the nomenclature of~\cite[p.~37]{ChEn05}.

Next we show~\eqref{I_infty}. We have
\begin{equation*}
	I_\infty = \int_a^\infty \sinh(s)^{d-1} \Bigl(\int_s^\infty \sinh(s)^{-(d-1)}\Bigr)\,ds.
\end{equation*}
With the inequalities
\begin{equation}	\label{sinh_infty}
	c_a\, e^s \le \sinh(s) \le \frac{1}{2}\,e^s,\qquad s\in [a,+\infty),
\end{equation}
where $c_a = (1-\exp(-2a))/2$, we obtain (recall that $d\ge 2$)
\begin{align*}
	I_\infty &\ge {\rm const.} \int_a^\infty e^{(d-1)s} \Bigl(\int_s^\infty e^{-(d-1)u}\,du\Bigr)\,ds\\
			 &=   {\rm const.} \int_a^\infty e^{(d-1)s}\,e^{-(d-1)s}\,ds\\
			 & =+\infty,
\end{align*}
and~\eqref{I_infty} is proved. Together with~\eqref{rho_a} this result shows
that the hypotheses of theorem~4.2 in~\cite{ChEn05} are satisfied. This entails
that for every start point $s_0\ge a$ there exists a unique weak solution of~\eqref{SDE_S}
up to the first hitting time $T_a$, and the solution is transient in the sense
stated in the lemma. Moreover, we obtain that the behavior of the SDE~\eqref{SDE_S}
at infinity is of \emph{type B} as defined in~\cite[p.~82]{ChEn05}.

Finally, with the result that the SDE~\eqref{SDE_S} has a (right) singularity of 
type~3 and the behavior of type~B at infinity, we can apply theorem~4.6.(viii)
in~\cite{ChEn05}, which implies the statement of the lemma. (We remark that the
statement of the above quoted theorem in~\cite{ChEn05} is formulated there for a 
two-sided singularity, but actually the properties of the  SDE on the negative half 
axis do not enter the statement nor its proof at all. In order to bring our situation 
precisely into the one discussed in chapter~4 of~\cite{ChEn05}, we simply could 
interpret the SDE~\eqref{SDE_S} as one formulated on all of $\R\setminus\{0\}$, and 
we would get the same result.) 
\end{proof}

Next we show

\begin{lemma}	\label{lem_punique}
The solutions of the SDE~\eqref{SDE_S} are pathwise 
unique.
\end{lemma}

\begin{proof}
This statement is a direct consequence of the fact that the drift $x\mapsto 
b_0(x) = \coth(x)$ is decreasing on $(0,+\infty)$, cf.\ Example~5.2.4 in~\cite{KaSh91}.
\end{proof}

Lemma~\ref{lem_punique} allows the application of the Yamada--Watanabe theorem~\cite{YaWa71} 
(cf.\ also~\cite[Theorem~1.1, Chap.~IV]{IkWa89} or \cite[Theorem~1.7, Chap.~IX]{ReYo91}) 
which entails that we even have strong solutions, and thereby concludes the proof
of theorem~\ref{thm_ex_un_b0}.

Having established the existence of Wiener processes on the mass shells $\cV_m^d$, 
$d\ge 2$, we now turn to the special cases $d=2$, $d=3$, and provide more explicit 
expressions descriptions of theses processes in terms of stochastic differential 
equations in hyperbolic as well as cartesian coordinates.

%
%

\subsection{The Case \boldmath $d=2$}	\label{ssect_d_equal_2}
Consider formula~\eqref{LB2} for the Laplacian on $\cV_m^2$. Thus the associated
It\^o stochastic differential equations for stochastic processes $S$, and $\Phi$ 
in the $s$, resp.\ $\vp$ coordinates are
\begin{equation}	\label{SDE_d_2}
\begin{split}
	dS_t    &= \frac{1}{2m^2}\,\coth(S_t)\,dt + \frac{1}{m}\,dW^1_t\\
	d\Phi_t &= \frac{1}{m\sinh(S_t)}\,dW^2_t,
\end{split}	
\end{equation}
where $W^1$ and $W^2$ are independent standard one dimensional Wiener processes.
Of course, the solutions of the SDE for $\Phi$ have to be taken modulo~$2\pi$.
We transform these equations into three dimensional cartesian coordinates
using It\^o calculus. A straightforward computation with It\^o's formula yields
the following stochastic differential equations for the cartesian components 
$P=(P_0, P_1, P_2)$ 
\begin{equation}	\label{CSDE_2}
\begin{split}
	dP_0(t) &= \frac{P_0(t)}{m^2}\,dt + \frac{r(t)}{m}\,dW^1_t\\[1ex]
	dP_1(t) &= \frac{P_1(t)}{m^2}\,dt + \frac{P_0(t)P_1(t)}{mr(t)}\,dW^1_t 
							- \frac{P_2(t)}{r(t)}\,dW^2_t\\[1ex]
	dP_2(t) &= \frac{P_2(t)}{m^2}\,dt + \frac{P_0(t)P_2(t)}{mr(t)}\,dW^1_t + \frac{P_1(t)}{r(t)}\,dW^2_t,
\end{split}
\end{equation}
where we have set $r(t) = \sqrt{\mathstrut P_1(t)^2+P_2(t)^2}$. 
An application of It\^o calculus yields the associated generator in cartesian 
coordinates acting on smooth functions on~$\R^3$:
\begin{equation}	\label{eqGenCart_2}
\begin{split}
	L_2 = \frac{1}{2m^2}\Bigl((p_0^2-m^2)\,\p_0^2 + \sum_{i=1}^2 &(p_i^2+m^2)\,\p_i^2\\
				&+2\sum_{k>l=0}^2 p_k p_l \p_k\p_l+ 2 \sum_{k=0}^2 p_k \p_k\Bigr),
\end{split}
\end{equation}
where $\p_i$, $i=0$, $1$, $2$, denotes the usual partial derivative in the $i$--th 
coordinate direction. We want to point out the appearance of a first order term with  the 
linear ``drift'' coefficient function $\mb{p}\mapsto 1/m^2 \mb{p}$ in the generator~$L$.

\subsection{The Case \boldmath $d=3$}	\label{ssect_d_equal_3}
From the form~\eqref{LB3} of $\gD_3$ we deduce the following system of stochastic differential
equations for coordinate processes $S$, $\Theta$, $\Phi$:
\begin{equation}\label{SDE_d_3}
\begin{split}
	dS_t      &= \frac{1}{m^2}\,\coth(S_t)\,dt + \frac{1}{m}\,dW^1_t\\
	d\Theta_t &= \frac{1}{2m^2 \sinh(S_t)^2}\,\cot(\Theta_t)\,dt + \frac{1}{m\sinh(S_t)}\,dW^2_t\\
	d\Phi_t   &= \frac{1}{m\sinh(S_t)\sin(\Theta_t)}\,dW^3_t,
\end{split}
\end{equation}
where $W^1$, $W^2$, and $W^3$ are independent a standard one dimensional Wiener processes.
It is clear that also here the solutions of the equation for $\Phi$ have to taken modulo~$2\pi$.
A straightforward --- even though somewhat lengthy --- calculation with It\^o's formula
gives the following stochastic differential equations in cartesian coordinates of $\R^4$:
\begin{equation}	\label{CSDE_3}
\begin{split}
	dP_0(t) &= \frac{3}{2m^2}\,P_0(t)\,dt + \frac{R(t)}{m}\,dW^1_t\\
	dP_1(t) &= \frac{3}{2m^2}\,P_1(t)\,dt + \frac{P_0(t) P_1(t)}{m R(t)}\,dW^1_t\\ 
			&\hspace{8em} + \frac{P_1(t) P_3(t)}{r(t) R(t)}\,dW^2_t- \frac{P_2(t)}{r(t)}\,dW^3_t\\
	dP_2(t) &= \frac{3}{2m^2}\,P_2(t)\,dt + \frac{P_0(t) P_2(t)}{m R(t)}\,dW^1_t\\ 
			&\hspace{8em} + \frac{P_2(t) P_3(t)}{r(t) R(t)}\,dW^2_t + \frac{P_1(t)}{r(t)}\,dW^3_t\\
	dP_3(t) &= \frac{3}{2m^2}\,P_3(t)\,dt + \frac{P_0(t) P_3(t)}{m R(t)}\,dW^1_t - \frac{r(t)}{R(t)}\,dW^2_t.
\end{split}
\end{equation}
In the last equations we have set $R(t) = \sqrt{\mathstrut P_1(t)^2 + P_2(t)^2 + P_3(t)^2}$,
and $r(t)$ is as above. The generator has in cartesian coordinates the following form
\begin{equation}	\label{eqGenCart_3}
\begin{split}
	L_3 = \frac{1}{2m^2}\Bigl((p_0^2-m^2)\,\p_0^2 + \sum_{i=1}^3 &(p_i^2+m^2)\,\p_i^2\\
				&+2\sum_{k>l=0}^3 p_k p_l \p_k\p_l+ 3 \sum_{k=0}^3 p_k \p_k\Bigr),
\end{split}
\end{equation}
and also in this case we remark the linear drift term with a linear coefficient function 
$\mb{p}\mapsto 3/2 m^2 \mb{p}$.

\section{Relativistic Ornstein--Uhlenbeck Process}	\label{sect_OU}
Based on the Wiener process constructed on the mass shell $\cV_m^d$ in the previous section, 
we shall construct here stochastic processes on $\cV^d_m$ which resemble the standard 
Ornstein--Uhlenbeck process. As we have recalled in section~\ref{sect_Intro}, in the usual 
euclidean setting the Ornstein--Uhlenbeck process is constructed by adding (via a stochastic 
differential equation) a linear drift term to a standard Wiener process, which pushes the 
Wiener process back towards the origin. As a consequence, the classical Ornstein--Uhlenbeck 
process has an invariant distribution which is given
by a centered normal law. 

Consider first the special cases $d=2$, $3$, and the SDE's ~\eqref{CSDE_2}, \eqref{CSDE_3}, 
for the Wiener processes on the mass shell. Clearly, one cannot simply add linear drift
terms to these SDE's, since there is no guarantee that the resulting process would continue
to live on the mass shells --- actually, as our computations below show, this will definitely 
not be the case. 

Instead we introduce --- for general space dimension $d$ --- an additional drift term into the 
SDE~\eqref{SDE_S} in hyperbolic coordinates. Then we take the skew product of this new process
in the $s$-coordinate with a standard Wiener process on the unit sphere $S^{d-1}$ 
as in~\eqref{tau}, \eqref{skewpr}. Transforming this process via~\eqref{hyp_coord} into 
a stochastic process with values in $\R^{1+d}$ we obtain a process which lives on 
the mass shell $\cV^d_m$, if started thereon.

It turns out that a simple, natural choice for the additional drift term is given
by $s\mapsto -\gg/2 m^2\,\tanh(s)$, where $\gg$ is some non-negative constant. So we consider 
now the SDE
\begin{equation}	\label{SDE_tanh}
	dS_t = b_\gg(S_t)\,dt + \frac{1}{m^2}\,dW_t,\quad S_0 =s_0\in (0,+\infty),\qquad t\in\R_+,
\end{equation}
with
\begin{equation}	\label{drift}
	b_\gg(s) = \frac{d-1}{2m^2}\coth(s)- \frac{\gg}{2m^2} \tanh(s),\qquad s\in (0,+\infty).
\end{equation}
Hence for $\gg>d-1$ we have a backward drift which is asymptotically constant with 
value~$(d-\gg-1)/2m^2$. The choice of this additional drift term has two advantages: For one,
it turns out that in the special cases $d=2$, $d=3$ the SDE's in cartesian coordinates
will be supplemented with almost linear drifts, which are directed towards the origin and 
compensate the linear outward drifts which we had observed for Wiener processes in the 
SDE's~\eqref{CSDE_2}, \eqref{CSDE_3}. Therefore this shows some similarity with the 
construction of the classical Ornstein--Uhlenbeck process. Moreover, for $\gg$ large enough 
this additional drift yields the existence of an invariant state for the resulting process, 
which can be computed explicitly (as well as some other invariant states, see 
section~\ref{sect_IM}).

For the question of existence and uniqueness of solutions of~\eqref{SDE_tanh} we have 
the following

\begin{theorem}\label{thm_ex_un_OU}
For every initial condition $S_0=s_0\in (0,+\infty)$, the stochastic differential 
equation~\eqref{SDE_tanh} has a pathwise unique, strong solution which is a.s.\ strictly 
positive for all times. For $\gg\in[0,d-1)$ the solution is transient in the same sense
as in theorem~\ref{thm_ex_un_b0}. For $\gg\ge d-1$ the solution is recurrent in the sense 
that if $a>0$ and $s_0>a$ then $P$--a.s.\ $T_a<+\infty$.
\end{theorem}

\begin{proof} The proof is quite similar to the one of theorem~\ref{thm_ex_un_b0}, so we 
only sketch it. Again we temporarily put $m^2=1$ for notational simplicity. In this case 
the scale density $\rho$ becomes
\begin{equation*}
	\rho(s) = \frac{\sinh(a)^{d-1}}{\cosh(a)^\gg}\,\frac{\cosh(s)^\gg}{\sinh(s)^{d-1}},\qquad
						s\in (0,+\infty).
\end{equation*}
Therefore the estimations leading to the inequalities~\eqref{rho_a} and~\eqref{I_a}
are completely unaffected by the additional smooth, bounded function $s\mapsto \cosh(s)$, and
we find again that the singularity at $s=0$ is of \emph{type~3}. However, instead 
of~\eqref{rho_infty} we this time get the following for any fixed $a>0$:
\begin{equation}\label{rho_gg}
	\int_a^\infty \rho(s)\,ds\quad 
		\begin{cases}
			&< +\infty,\quad\text{if $\gg<   d-1$},\\
			&= +\infty,\quad\text{if $\gg\ge d-1$}.
		\end{cases}
\end{equation}
This shows that for $\gg\ge d-1$ we get now \emph{type~A} for the behavior
at $+\infty$, as defined on p.~82 in~\cite{ChEn05}. For $\gg\in[0,d-1)$ we have
to estimate the present analogue of $I_\infty$, see~\eqref{I_infty}. To this end, we use 
in addition to~\eqref{sinh} the trivial bounds $1/2\,\exp(s)\le \cosh(s)\le \exp(s)$. The
result is $I_\infty=+\infty$. Hence for $\gg\in[0,d-1)$ the behavior at infinity
is again of \emph{type~B}. Now we apply once more theorem~4.6.(viii) in~\cite{ChEn05} 
to conclude that for every initial condition $S_0=s_0>0$ we have the existence of a 
strictly positive weak solution which is unique in law.

Observe that the drift $b_\gg$ is monotone decreasing on $(0,+\infty)$, so that by the 
same argument as in the proof of lemma~\ref{lem_punique} pathwise uniqueness of the 
solutions holds true. Another application of the Ya\-ma\-da--Watanabe theorem provides 
us with the existence of a strong solution for every initial condition 
$S_0=s_0\in (0,+\infty)$.

Finally we remark that theorem~4.1, \cite[p.~81]{ChEn05}, states that the behavior 
of type~A at infinity of the SDE entails that the solutions are recurrent in the sense
of the theorem.
\end{proof}

As in section~\ref{sect_BM}, let $\Theta=(\Theta_t,\,t\in\R_+)$ be a standard Wiener 
process on the $d-1$ dimensional unit sphere $S^{d-1}$, and define the stochastic time 
scale~$\tau$ as in~\eqref{tau} where this time we choose for $S$ the process
defined by the SDE~\eqref{SDE_tanh}. Consider the skew product 
\begin{equation}	\label{ROU_hyp}
	\bigl((S_t,\Theta_{\tau(t)}),\,t\in\R_+\bigr).
\end{equation}
We transform this process with the equations~\eqref{hyp_coord} into a
stochastic process $\mb{P} = \bigl(\mb{P}(t),\,t\in\R_+\bigr)$, $\mb{P}(t) =
\bigl(P_0(t),P(t)\bigr)$, on the mass shell $\cV^d_m$ written in cartesian coordinates:
\begin{equation}	\label{ROUM}
\begin{split}
	P_0(t) &= m\,\cosh(S_t),\\
	P(t)   &= m\,\sinh(S_t)\,\go(\Theta_{\tau(t)}).
\end{split}
\end{equation}
We call the process $\mb{P}$ the \emph{relativistic Ornstein--Uhlenbeck momentum process}
in $1+d$ dimensions. The \emph{relativistic Ornstein--Uhlenbeck velocity process}
$V=\bigl(V(t),\,t\in\R_+\bigr)$ in $1+d$ dimensions is then defined as
\begin{equation} \label{ROUV}
	V(t) = \frac{P(t)}{P_0(t)} = \tanh(S_t)\,\go(\Theta_{\tau(t)}),\qquad t\in\R_+.
\end{equation}
(Recall that we work with physical units so that the speed of light $c$ in the vacuum 
is equal to $1$. In other units, we have an additional factor $c$ on the right hand side.)

Similarly as for the Wiener process, which we treated in section~\ref{sect_BM},
for the cases of dimensions $d=2$ and $d=3$, we give an alternative, more 
explicit description of the relativistic Ornstein--Uhlenbeck processes in cartesian 
coordinates in terms of stochastic differential equations instead of using the
skew product.

For $d=2$ we replace the first equation in~\eqref{SDE_d_2} by~\eqref{SDE_tanh}
and transform them into cartesian coordinates with a straightforward computation 
using It\^o's formula. This yields the following SDE's for the components of $\mb{P}$
\begin{equation}	\label{CSDE_2_drift}
\begin{split}
	dP_0(t) &= \frac{1}{2m^2}\,(2-\gg)P_0(t)\,dt +\frac{\gg}{2P_0(t)}\,dt
					+ \frac{r(t)}{m}\,dW^1_t\\[1ex]
	dP_1(t) &=  \frac{1}{2m^2}\,(2-\gg)P_1(t)\,dt + \frac{P_0(t)P_1(t)}{mr(t)}\,dW^1_t 
					- \frac{P_2(t)}{r(t)}\,dW^2_t\\[1ex]
	dP_2(t) &= \frac{1}{2m^2}\,(2-\gg)P_2(t)\,dt + \frac{P_0(t)P_2(t)}{mr(t)}\,dW^1_t 
					+ \frac{P_1(t)}{r(t)}\,dW^2_t,
\end{split}
\end{equation}
where we have set $r(t) = \sqrt{\mathstrut P_1(t)^2 + P_2(t)^2}$. Thus, for $\gg\ge 2$
the original outward drift of the Wiener process is compensated, while for $\gg>2$ 
we have an effective drift towards the origin, and except for the term $\gg/2P_0(t)\,dt$ this 
drift acts in a linear way as for the classical Ornstein--Uhlenbeck process. The additional 
non-linear term in the equation for $P_0$ takes care that the process stays 
on the mass shell. Note however, that this term is bounded from above by $\gg/2m$ since 
$P_0(t)\ge m$ on the mass shell. 

For $d=3$ we obtain
\begin{equation}	\label{CSDE_3_drift}
\begin{split}
	dP_0(t) &= \frac{1}{2m^2}\,(3-\gg)P_0(t)\,dt + \frac{\gg}{2P_0}\,dt 
					+ \frac{R(t)}{m}\,dW^1_t\\
	dP_1(t) &= \frac{1}{2m^2}\,(3-\gg)P_1(t)\,dt + \frac{P_0(t) X_1(t)}{m R(t)}\,dW^1_t \\ 
			&\hspace{8em}+ \frac{P_1(t) P_3(t)}{r(t) R(t)}\,dW^2_t- \frac{P_2(t)}{r(t)}\,dW^3_t\\
	dP_2(t) &= \frac{1}{2m^2}\,(3-\gg)P_2(t)\,dt + \frac{P_0(t) P_2(t)}{m R(t)}\,dW^1_t \\ 
			&\hspace{8em}+ \frac{P_2(t) P_3(t)}{r(t) R(t)}\,dW^2_t + \frac{P_1(t)}{r(t)}\,dW^3_t\\
	dP_3(t) &= \frac{1}{2m^2}\,(3-\gg)P_3(t)\,dt + \frac{P_0(t) P_3(t)}{m R(t)}\,dW^1_t 
				- \frac{r(t)}{R(t)}\,dW^2_t,
\end{split}
\end{equation}
where $R(t) = \sqrt{\mathstrut P_1(t)^2 + P_2(t)^2 + P_3(t)^2}$, and $r(t)$ is as in
the case $d=2$ above. So in this case we have to have $\gg\ge 3$ in order to compensate 
the outward drift of the Wiener process, and for $\gg>3$ we have as above an almost linear 
drift term pushing the motion towards the origin.

The generators of these processes are those obtained for the Wiener process plus
the additional drift terms derived above, namely for $d=2$
\begin{equation}	\label{eqGenCart_2_drift}
	L_2 - \frac{\gg}{2m^2}\,\sum_{k=0}^2 p_k\p_k + \frac{\gg}{2p_0}\,\p_0,
\end{equation}
and for $d=3$
\begin{equation}	\label{eqGenCart_3_drift}
	L_3 - \frac{\gg}{2m^2}\,\sum_{k=0}^3 p_k\p_k + \frac{\gg}{2p_0}\,\p_0,
\end{equation}
where $L_2$ and $L_3$ are as in~\eqref{eqGenCart_2}, ~\eqref{eqGenCart_3} respectively.
Note that $p_0\ge m$ so that the non-linear drift coefficients $\gg/2p_0$ in the time direction are 
bounded from above by $\gg/2m$, and they asymptotically vanish as $p_0$ tends to~$+\infty$.

\section{Invariant Measures} \label{sect_IM}
Define a measure $\mu_d$ on $(\R_+,\cB(\R_+))$ by
\begin{equation}	\label{mu_d}
	\mu_d(ds) = \sinh(s)^{d-1}\,ds,\qquad s\in\R_+.
\end{equation}

\begin{lemma}	\label{lem_inv_meas_S}
For every $\gg>d-1$ and every initial condition $S_0$, the solution $S=(S_t,\,t\in\R_+)$ 
of the stochastic differential equation~\eqref{SDE_tanh} has the following invariant 
measure 
\begin{equation}	\label{eq_inv_meas_S}
	\frac{1}{N_{d,\gg}}\,\cosh(s)^{-\gg}\mu_d(ds),
\end{equation}
where
\begin{equation}	\label{eq_norm}
	N_{d,\gg} = \int_0^\infty \cosh(s)^{-\gg}\mu_d(ds).
\end{equation}
\end{lemma}

\begin{proof}
Consider the generator of $S=(S_t,\,t\in\R_+)$:
\begin{equation*}	\label{gen_S_drift}
	L_{d,\gg} = L_{d,0} - \frac{\gg}{2m^2}\,\tanh(s)\,\p_s,
\end{equation*}
with
\begin{equation*}	\label{gen_S}
\begin{split}
	L_{d,0} &= \frac{1}{2m^2\sinh(s)^{d-1}}\,\p_s \sinh(s)^{d-1}\,\p_s\\
			&= \frac{1}{2m^2}\bigl(\p_s^2 + (d-1) \coth(s)\,\p_s\bigr).
\end{split}	
\end{equation*}
Up to the factor $1/2$, $L_{d,0}$ is the part of the Laplace--Beltrami operator~\eqref{LBd}
involving the $s$--derivatives. Therefore, by the construction of the Laplace--Beltrami 
operator is symmetric with respect to the measure $\mu_d$ on $(\R_+,\cB(\R_+))$. Hence
the adjoint $L_{d,\gg}^{*}$ of $L_{d,\gg}$ with respect to $\mu_d$ acts on smooth functions
as the differential operator given by
\begin{equation*} 
	L_{d,\gg}^{*} 
		= L_{d,0} + \frac{\gg}{2m^2}\,\sinh(s)^{-(d-1)}\, \p_s \tanh(s)\,\p_s \sinh(s)^{d-1}.
\end{equation*}
Rewrite $L_{d,\gg}^{*}$ as follows
\begin{equation*}
	L_{d,\gg}^{*} 
		= \frac{1}{2m^2}\,\sinh(s)^{-(d-1)}\,\p_s\Bigl(\p_s - (d-1)\coth(s) 
					+ \gg \tanh(s)\Bigr)\sinh(s)^{d-1}.
\end{equation*}
An elementary computation shows that
\begin{equation*}
	\Bigl(\p_s - (d-1)\coth(s) + \gg \tanh(s)\Bigr)\sinh(s)^{d-1}\cosh(s)^{-\gg}=0,
\end{equation*}
finishing the proof.
\end{proof}

As long as $\gg>d-1$, we may equivalently consider the function 
\begin{equation*}
	s\mapsto N_{d,\gg}^{-1}\cosh(s)^{-\gg}\sinh(s)^{d-1}
\end{equation*} 
as the \emph{Lebesgue} density of the invariant measure
of the stochastic process $S=(S_t,\,t\in\R_+)$. This is in particular useful, when we
want to compare the theoretical result of lemma~\ref{lem_inv_meas_S} with simulations
of the process. Figure~\ref{figI} shows some of the results of simulation experiments
we carried out, and which are described in more technical detail in appendix~\ref{app_Sim}. 
In each of these experiments we have put $m^2=1$, and simulated $5\times 10^3$ (numerical 
approximations of) the paths of the process $S$ for a relatively long time (see 
appendix~\ref{app_Sim}), and plotted the resulting histograms of the final positions 
(in blue) versus the Lebesgue density (in red) derived above. The plots show a very 
reasonable agreement of the theoretical density with the histograms, as could be expected.

\begin{figure} \centering
\subfloat[$d=2$, $\gg=4$]
	{\includegraphics[scale=.3]{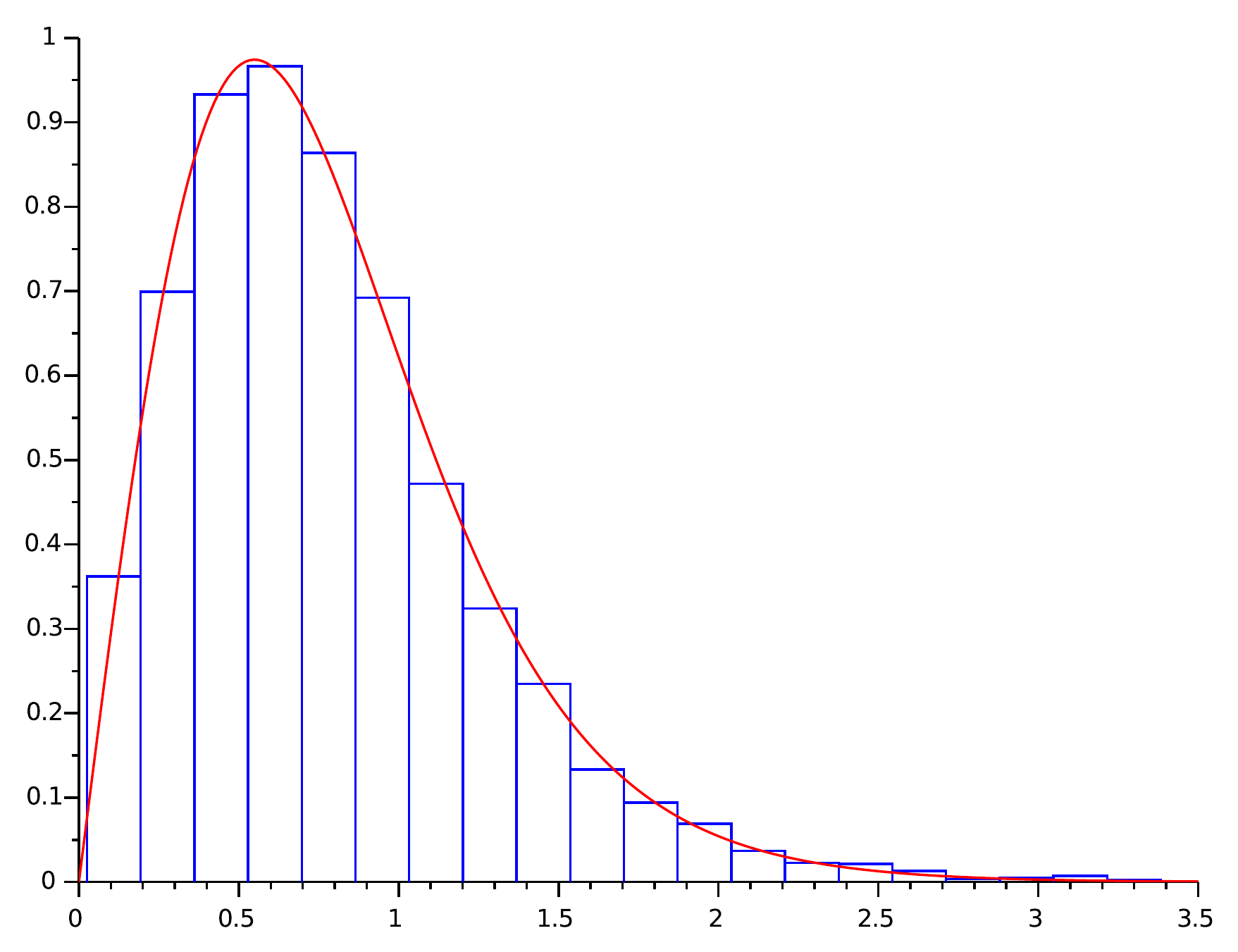}\label{figIa}}
\hspace{4em}
\subfloat[$d=3$, $\gg=4$]
	{\includegraphics[scale=.3]{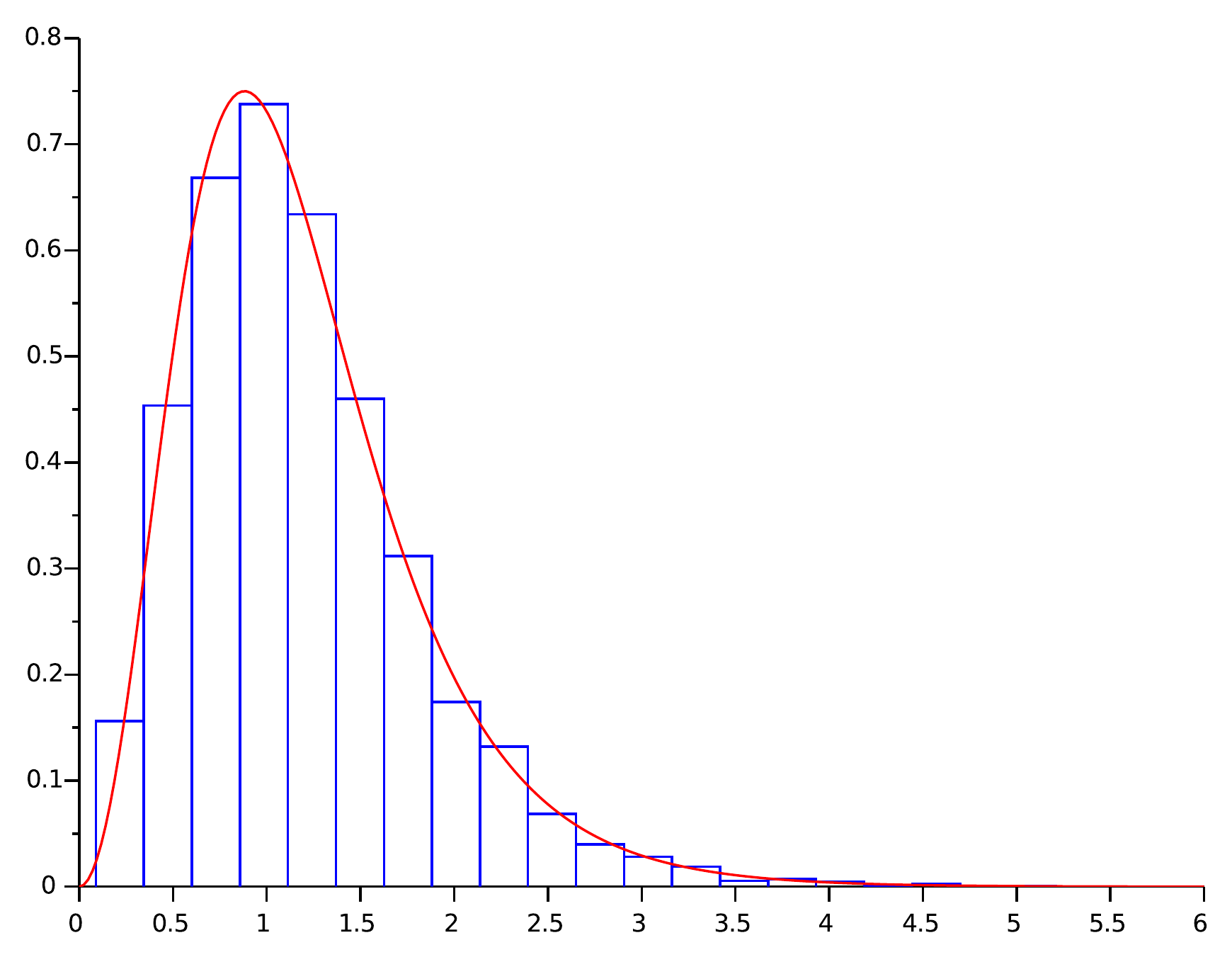}\label{figIb}}
\newline
\subfloat[$d=3$, $\gg=10$]
	{\includegraphics[scale=.3]{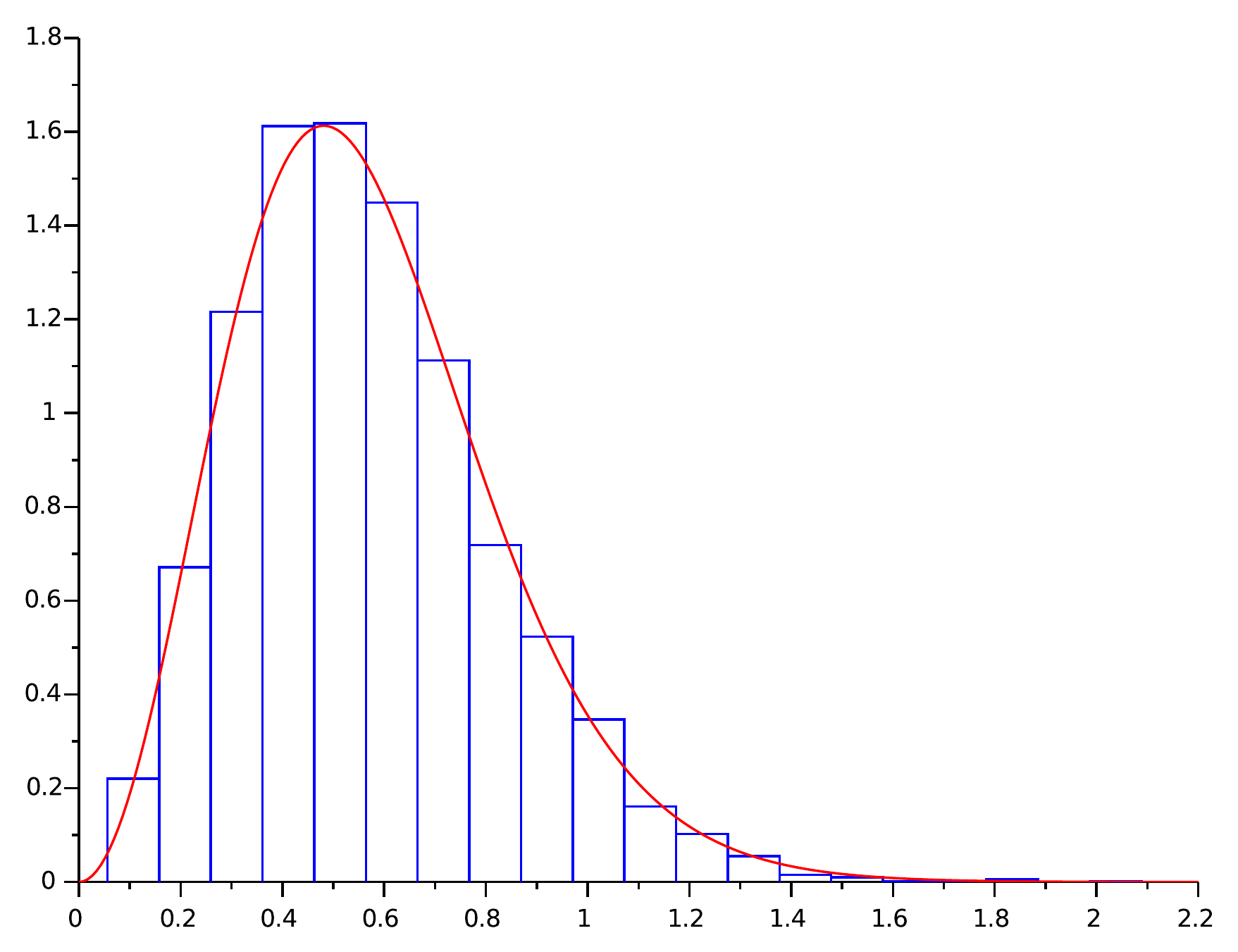}\label{figIc}}
\hspace{4em}
\subfloat[$d=4$, $\gg=7$]
	{\includegraphics[scale=.3]{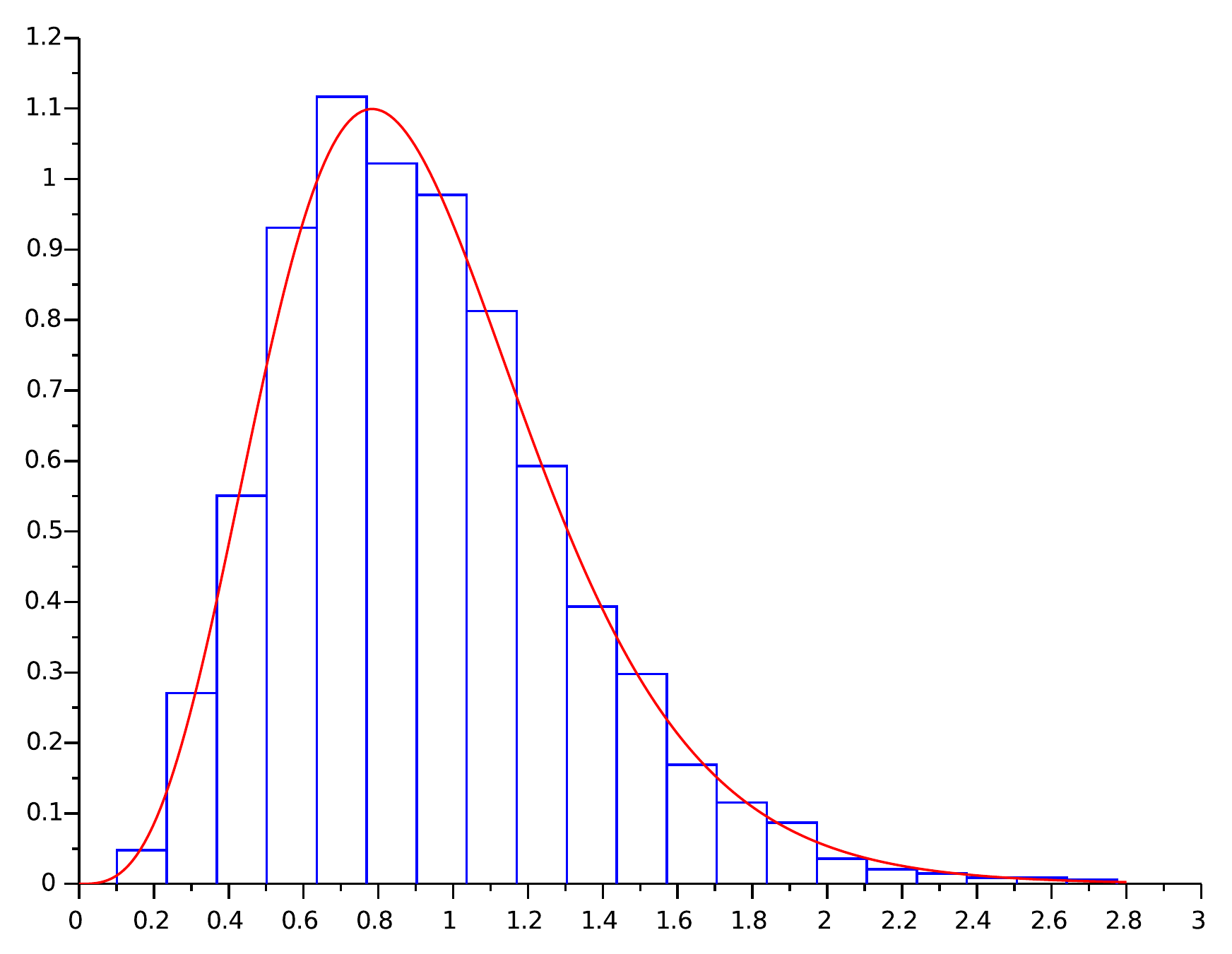}\label{figId}}
\caption{Histograms of Simulations of $S$ at Large Times}\label{figI}
\end{figure}

\begin{theorem}	\label{thm_inv_meas}
For every $\gg>d-1$, the stochastic process given by~\eqref{ROU_hyp} admits
the invariant measure given by
\begin{equation}	\label{eq_inv_meas_ROU}
	\frac{\gG(d/2)}{2\pi^{d/2} N_{d,\gg}}\,\cosh(s)^{-\gg}\,
		d\text{{\rm vol}}_d(s,\theta),\qquad s\in \R_+,\,\theta\in [0,\pi]^{d-2}\times[0,2\pi).
\end{equation}
\end{theorem}

\begin{proof}
This follows directly from lemma~\ref{lem_inv_meas_S}, together with the observation
that the invariant measure of the Wiener process on the unit sphere $S^{d-1}$ is the 
uniform law on $S^{d-1}$:
\begin{equation*}
	\frac{\gG(d/2)}{2 \pi^{d/2}}\,d\gs_{S^{d-1}}(\theta),
\end{equation*}
where the coefficient in front of the surface element $d\gs_{S^{d-1}}$ is the inverse 
of the total area of $S^{d-1}$.
\end{proof}

In a slightly informal manner the invariant measure~\eqref{eq_inv_meas_ROU} of
the relativistic Ornstein--Uhlenbeck momentum process can be written in
cartesian coordinates $\mb{p} = (p_0,p)\in\R^{1+d}$ as
\begin{equation}	\label{eq_inv_meas_cart}
\begin{split}
	\text{const.}\, \frac{m^\gg}{p_0^\gg}\,&\gd\bigl(p_0^2-|p|^2-m^2\bigr)\,
					1_{\R_+}(p_0)\,d^{1+d}\mb{p}\\
		&=\text{const.}\, \frac{m^\gg}{(m^2+|p|^2)^{\gg/2}}\,
					\gd\bigl(p_0^2-|p|^2-m^2\bigr)\,1_{\R_+}(p_0)\,d^{1+d}\mb{p},
\end{split}
\end{equation}
where the constant is the same as in~\eqref{eq_inv_meas_ROU}, and $\gd$ is the Dirac
delta function.

For simplicity let us put $m^2=1$ in the sequel. From lemma~\ref{lem_inv_meas_S} we 
directly get the following

\begin{corollary}	\label{cor_inv_meas_P0}
For $d\ge 2$ and $\gg>d-1$ the energy process $P_0$ has an invariant density $\vp_{P_0}$
with respect to Lebesgue measure on $(\R,\cB(\R))$ given by
\begin{equation}	\label{inv_meas_P0}
	\vp_{P_0}(p_0) 
		= \frac{1}{N_{d,\gg}}\,p_0^{-\gg}(p_0^2-1)^{(d-2)/2}\,
									1_{[1,+\infty)}(p_0),\qquad p_0\in\R.	
\end{equation}
The \emph{Ornstein--Uhenbeck speed process} $|V|=\sqrt{\mathstrut V_1^2 +\dotsb V_d^2}$ 
has an invariant density $\vp_{|V|}$ with respect to Lebesgue measure on $([0,1],\cB([0,1]))$ 
given by
\begin{equation}	\label{inv_meas_V}
	\vp_{|V|}(v) = \frac{1}{N_{d,\gg}} v^{d-1} (1-v^2)^{(\gg-(d+1))/2},\qquad v\in [0,1].
\end{equation}
\end{corollary}

\begin{remark}	\label{rem_V}
As formula~\eqref{inv_meas_V} shows, the parameter $\gg$ must actually be chosen to be
strictly larger than $d+1$ in order that the particle undergoing this process cannot
attain the speed of light with strictly positive probability (cf.\ also figure~\ref{figIIIa}).
\end{remark}

Figures~\ref{figII} and~\ref{figIII} show the long term histograms of the final values 
of $5\times 10^3$ simulated paths of $P_0$, $|V|$ respectively, in comparison with the 
marginal densities~\eqref{inv_meas_P0}, \eqref{inv_meas_V} respectively. 
\begin{figure} \centering
\subfloat[$d=3$, $\gg=4$]
	{\includegraphics[scale=.3]{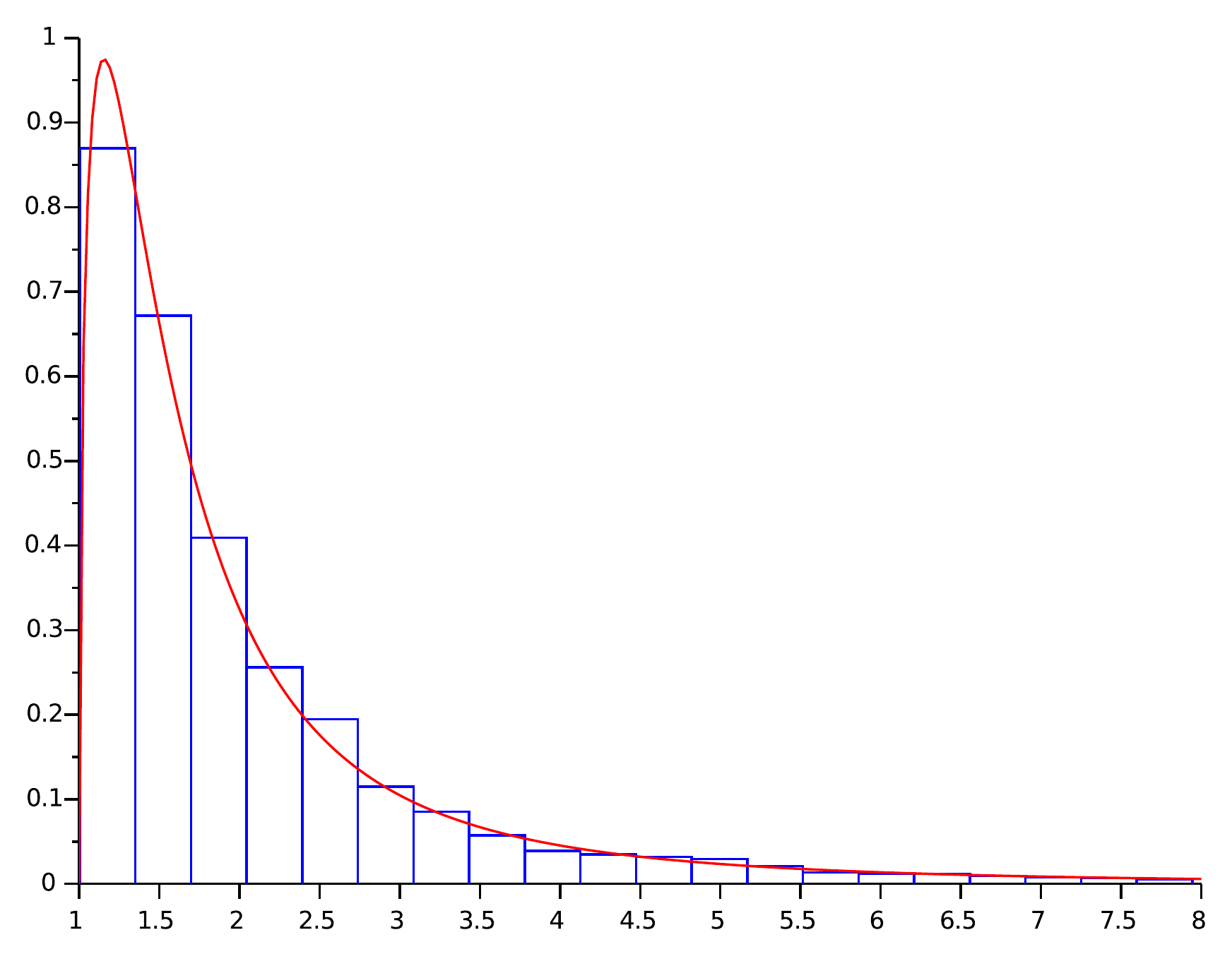}\label{figIIa}}
\hspace{4em}
\subfloat[$d=3$, $\gg=6$]
	{\includegraphics[scale=.3]{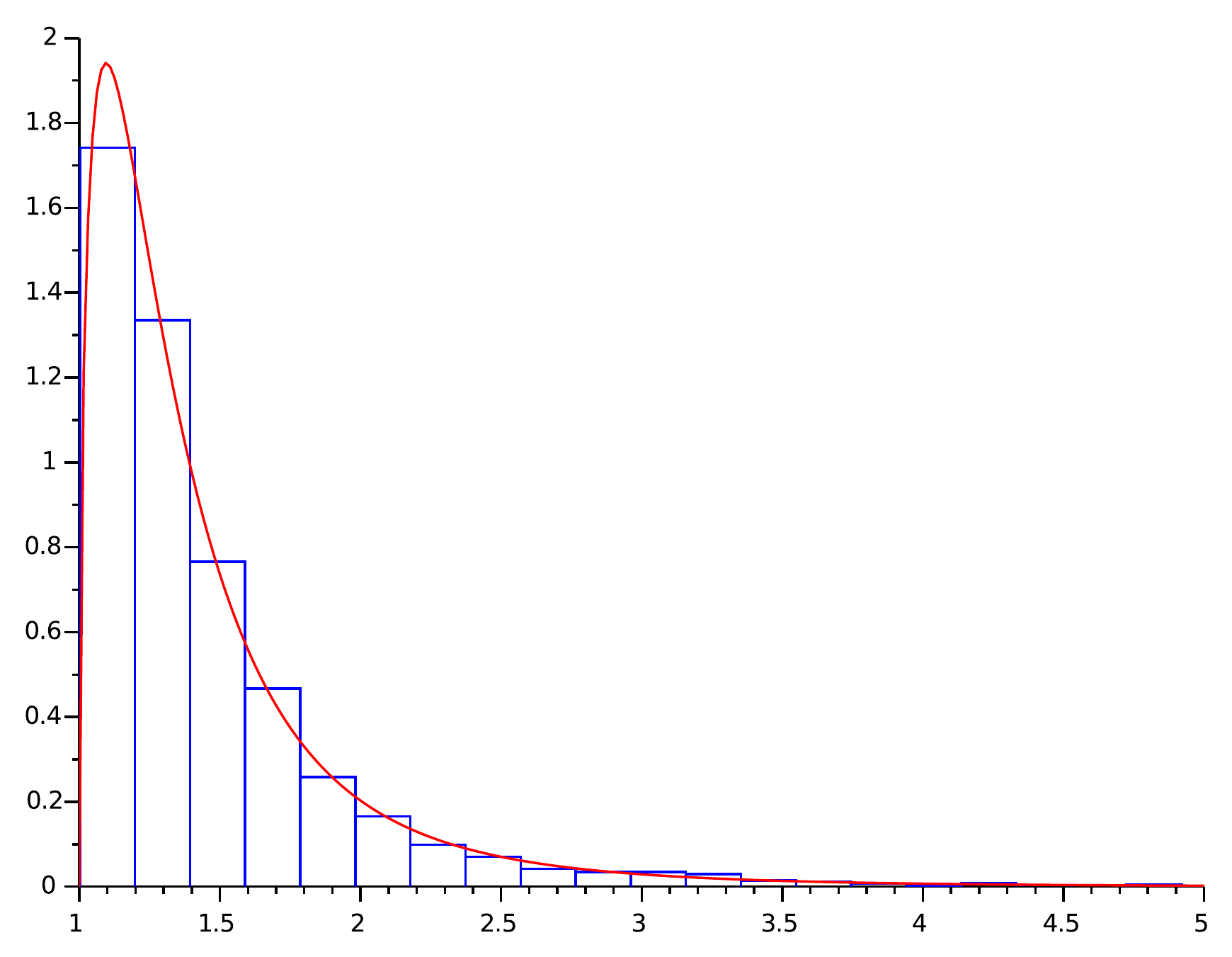}\label{figIIb}}
\newline
\subfloat[$d=3$, $\gg=8$]
	{\includegraphics[scale=.3]{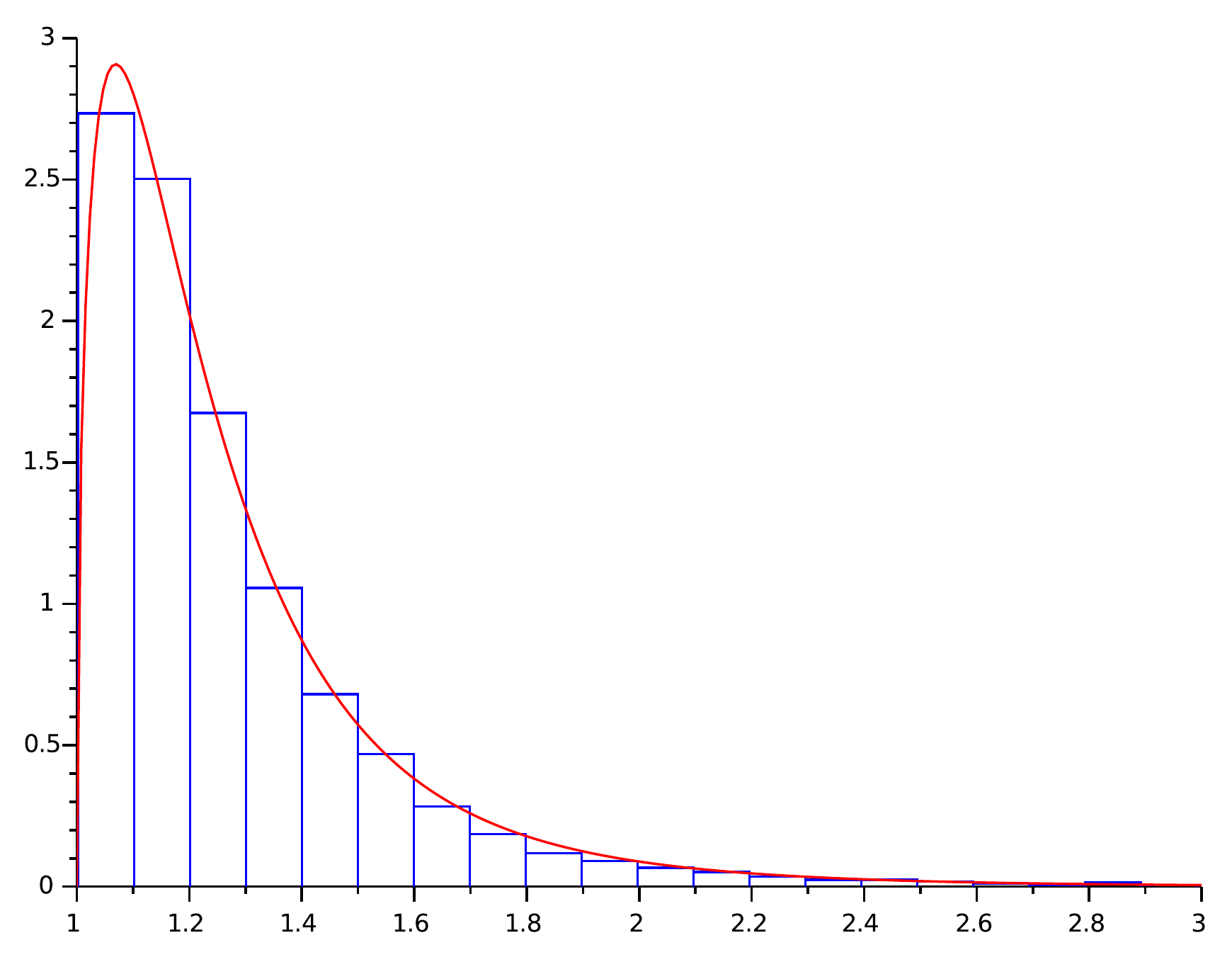}\label{figIIc}}
\hspace{4em}
\subfloat[$d=3$, $\gg=10$]
	{\includegraphics[scale=.3]{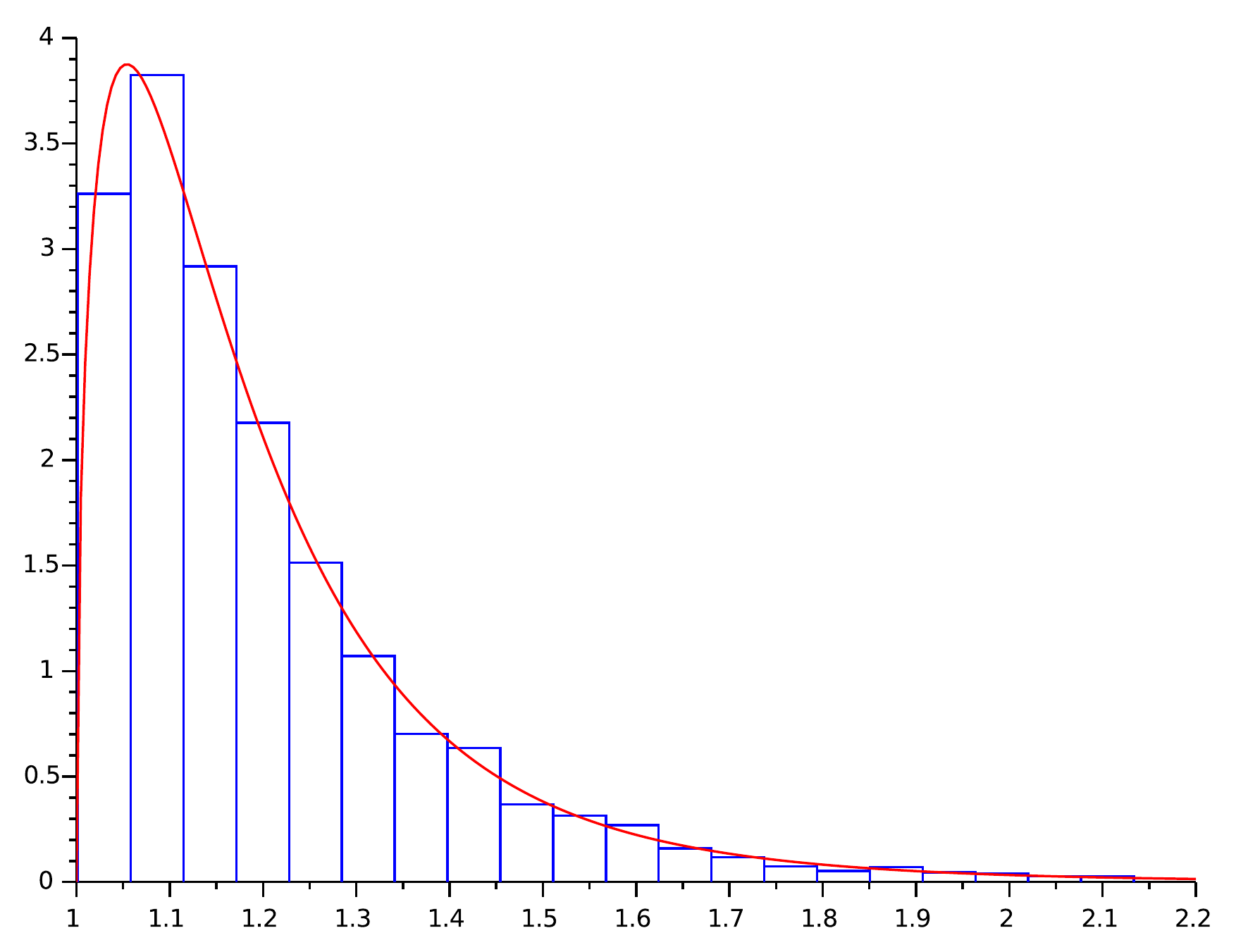}\label{figIId}}
\caption{Histograms of Simulations of $P_0$  at Large Times}\label{figII}
\end{figure}

\begin{figure} \centering
\subfloat[$d=3$, $\gg=4$]
	{\includegraphics[scale=.3]{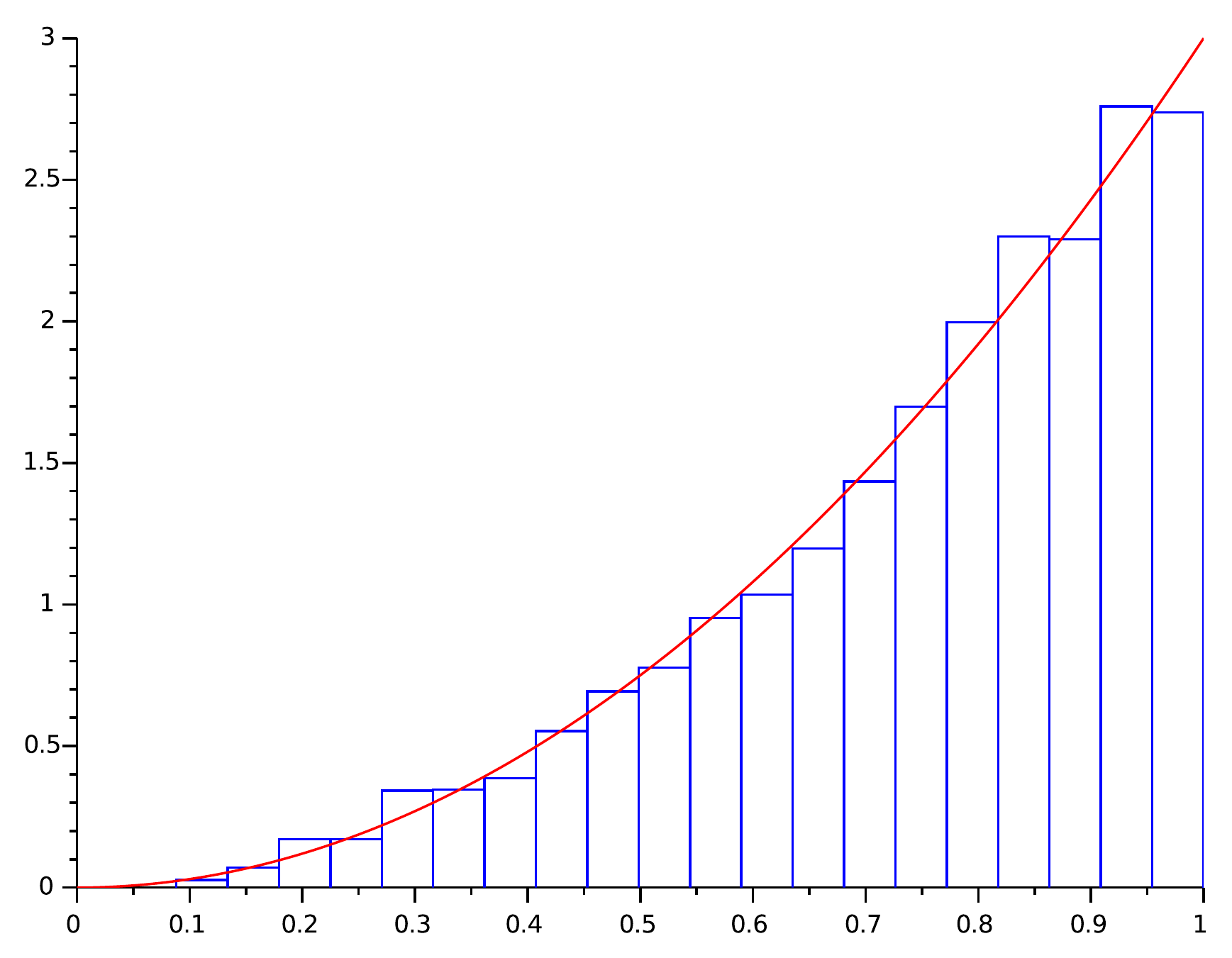}\label{figIIIa}}
\hspace{4em}
\subfloat[$d=3$, $\gg=6$]
	{\includegraphics[scale=.3]{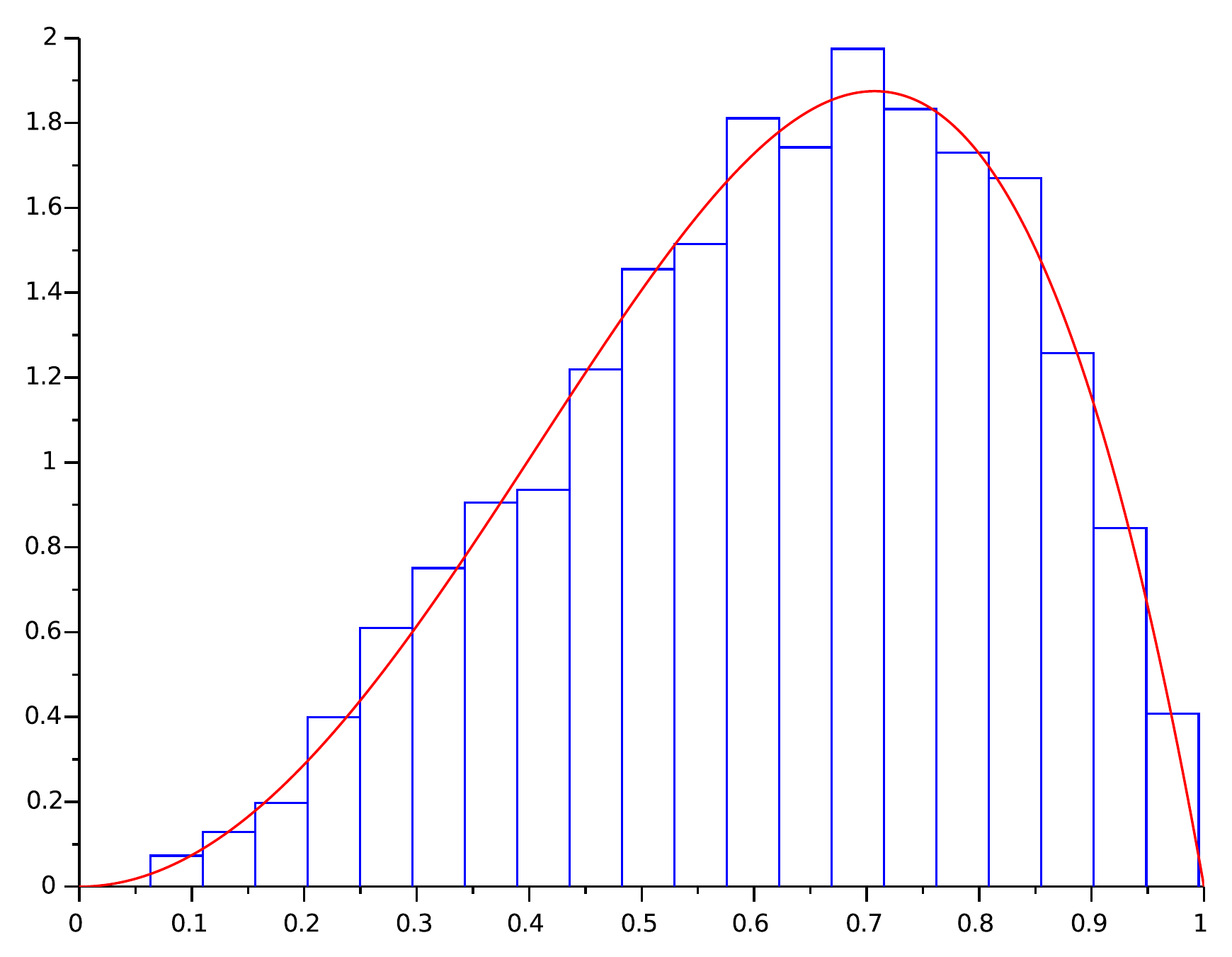}\label{figIIIb}}
\newline
\subfloat[$d=3$, $\gg=8$]
	{\includegraphics[scale=.3]{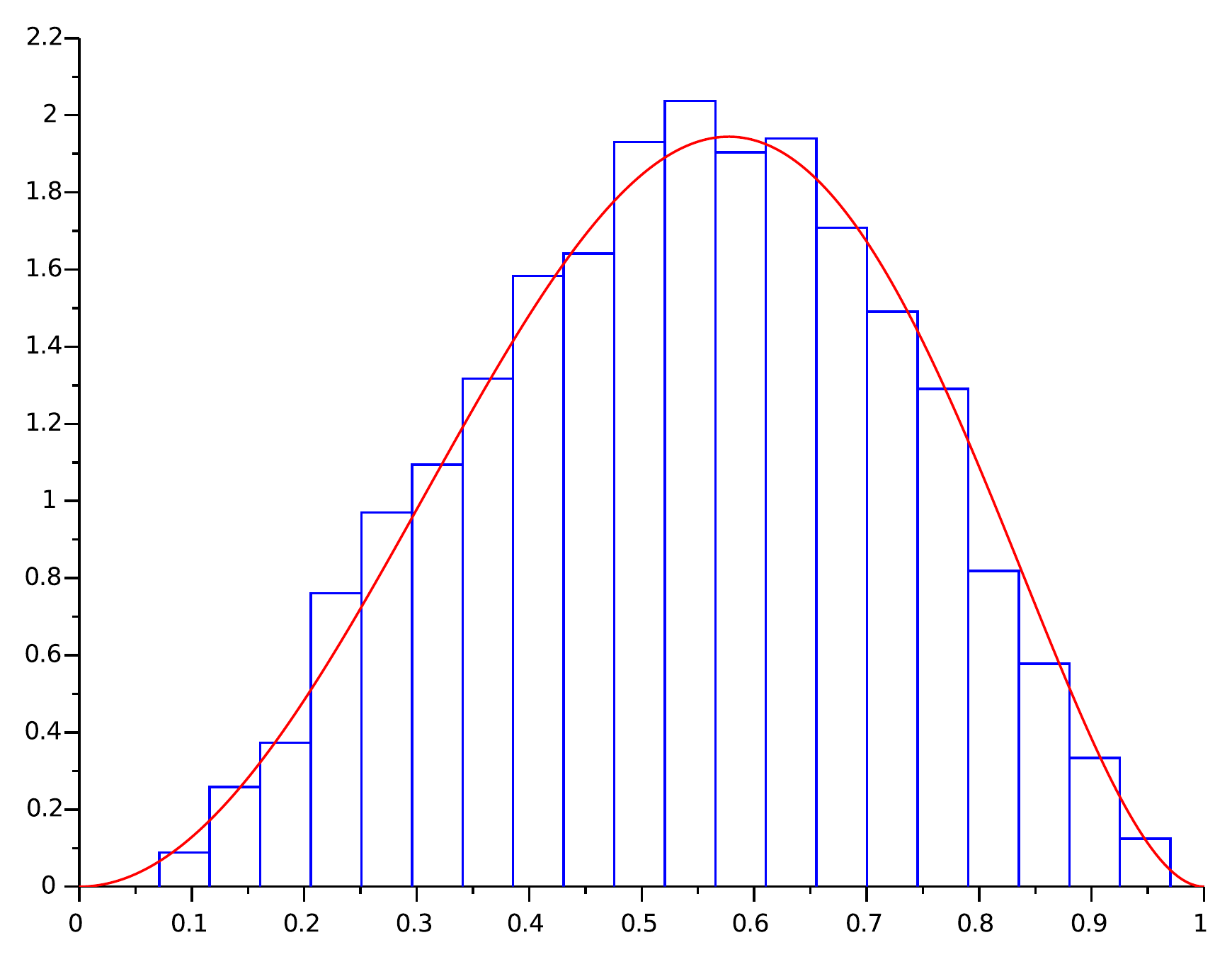}\label{figIIIc}}
\hspace{4em}
\subfloat[$d=3$, $\gg=10$]
	{\includegraphics[scale=.3]{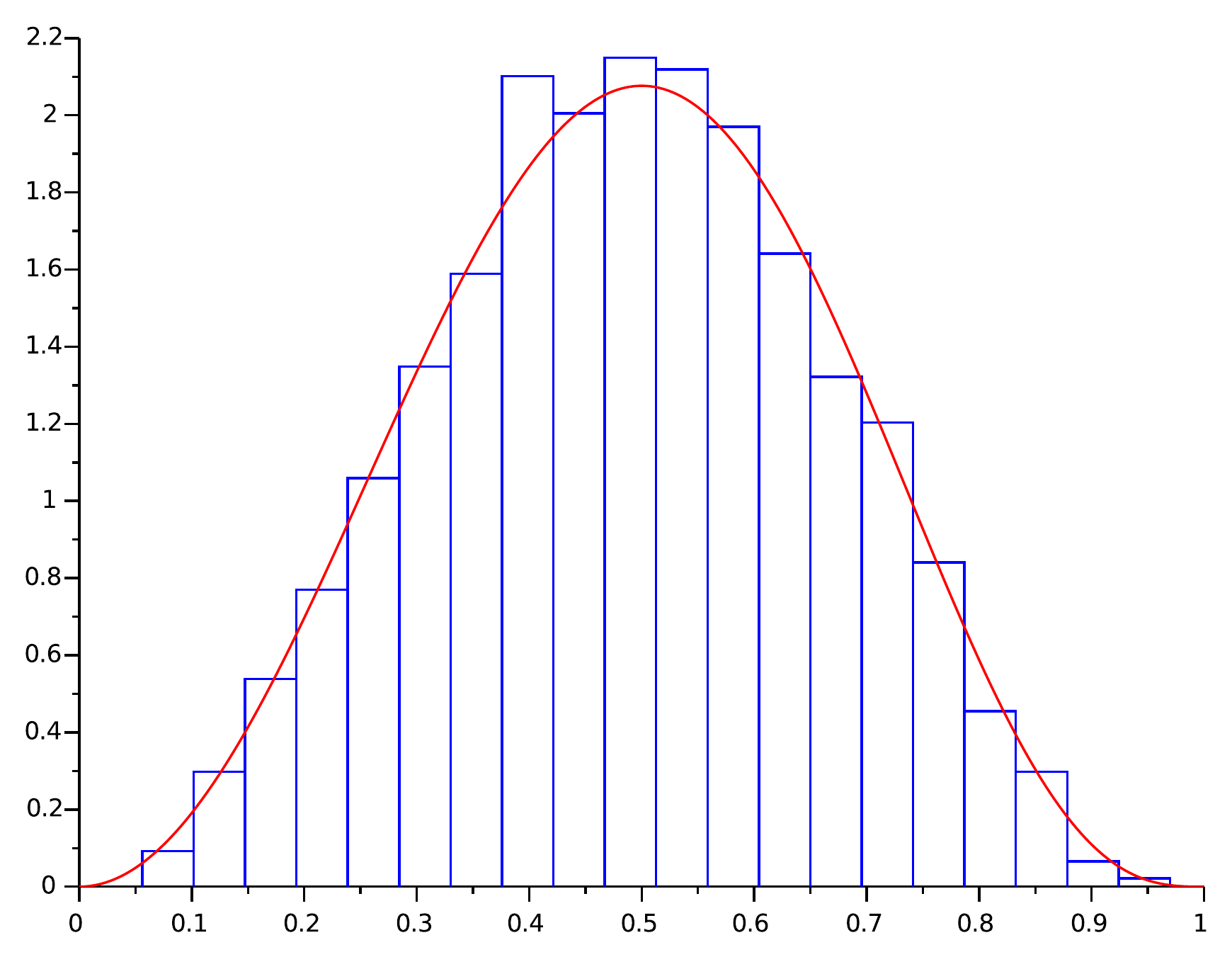}\label{figIIId}}
\caption{Histograms of Simulations of $|V|$  at Large Times}\label{figIII}
\end{figure}

For the remainder of this section we assume in addition that $d=3$, i.e., that we are
in the physical Minkowski space. Then it is straightforward to compute also the 
marginal invariant densities of the momentum processes $P_i$, $i=1$, $2$, $3$, explicitly:

\begin{corollary}	\label{cor_inv_meas_Pi}
For $d=3$ and $\gg>2$ the cartesian components $P_i$, $i=1$, $2$, $3$, of the
Ornstein--Uhlenbeck momentum process have marginal invariant densities $\vp_{P_i}$
with respect to Lebesgue measure on $(\R,\cB(\R))$ given by
\begin{equation}	\label{eq_mom_dens}
	\vp_{P_i}(p) = \frac{1}{n_\gg}\,(1+p^2)^{-(\gg-1)/2},\qquad i=1,2, 3,\,p\in\R,
\end{equation}
where $n_\gg$ is the normalization constant.
\end{corollary}

Figure~\ref{figIV} illustrates the result of corollary~\ref{cor_inv_meas_Pi} with
simulations of the value of $P_1$ for large times, $d=3$ and various values of~$\gg$.

\begin{figure} \centering
\subfloat[$d=3$, $\gg=4$]
	{\includegraphics[scale=.3]{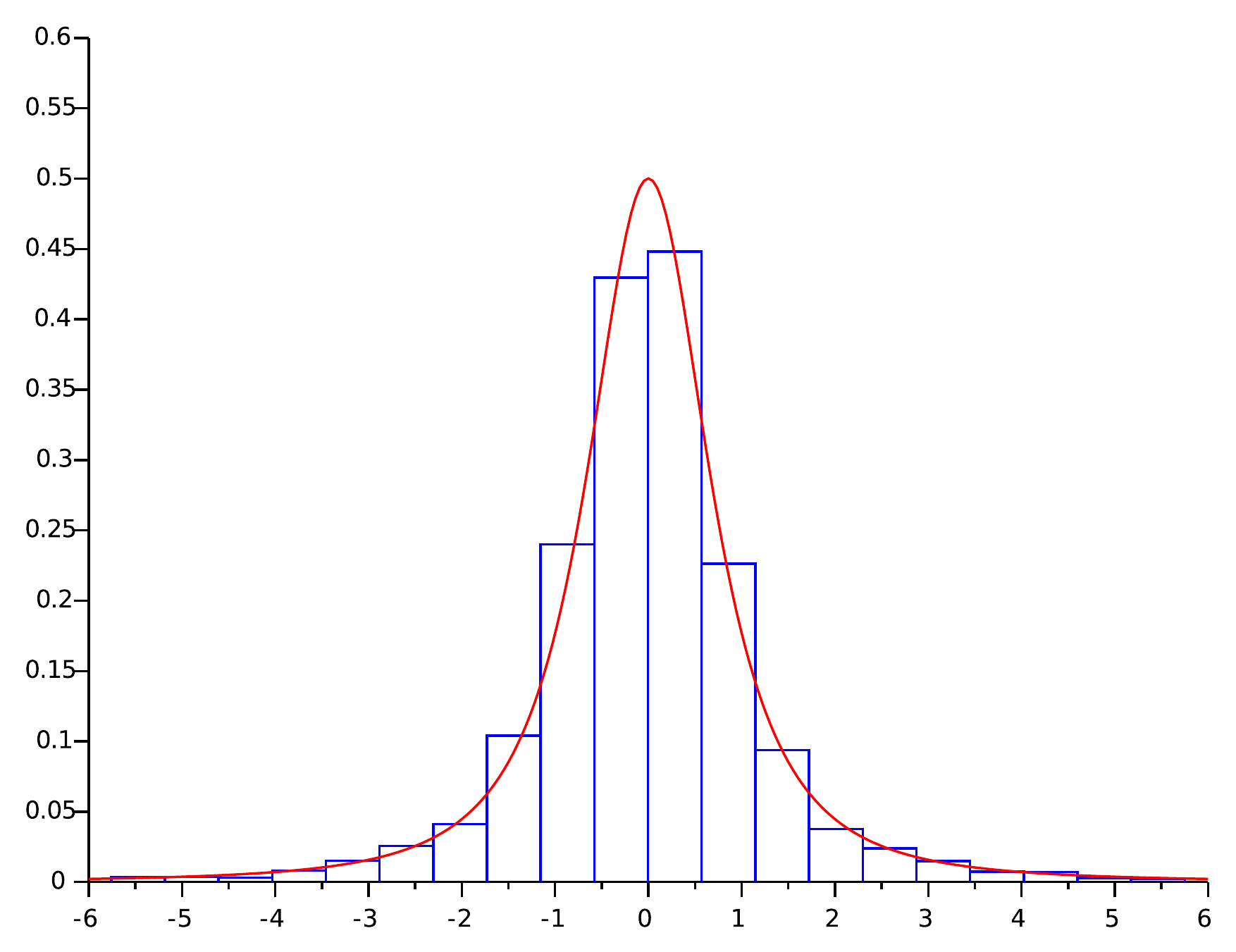}\label{figIVa}}
\hspace{4em}
\subfloat[$d=3$, $\gg=6$]
	{\includegraphics[scale=.3]{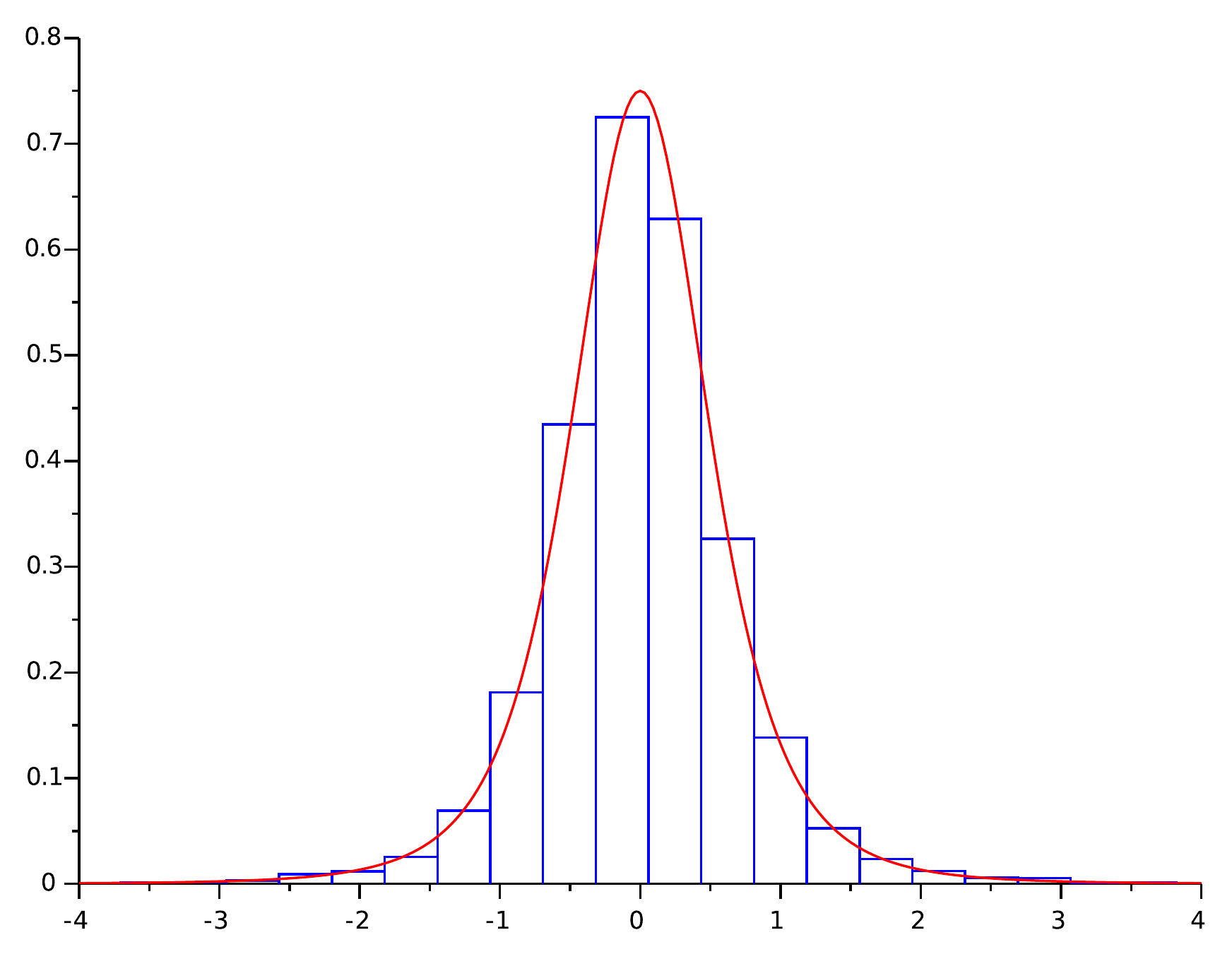}\label{figIVb}}
\newline
\subfloat[$d=3$, $\gg=8$]
	{\includegraphics[scale=.3]{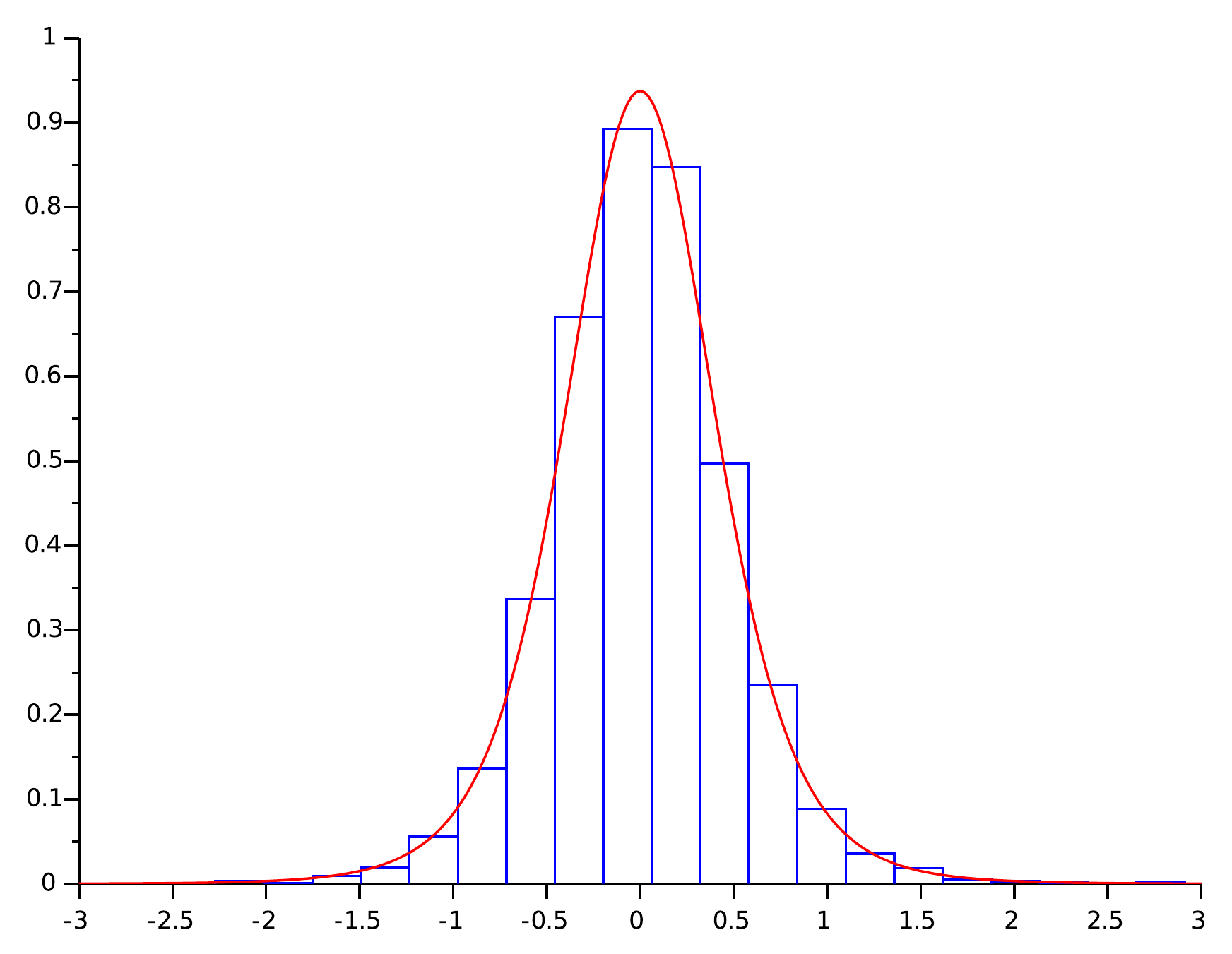}\label{figIVc}}
\hspace{4em}
\subfloat[$d=3$, $\gg=10$]
	{\includegraphics[scale=.3]{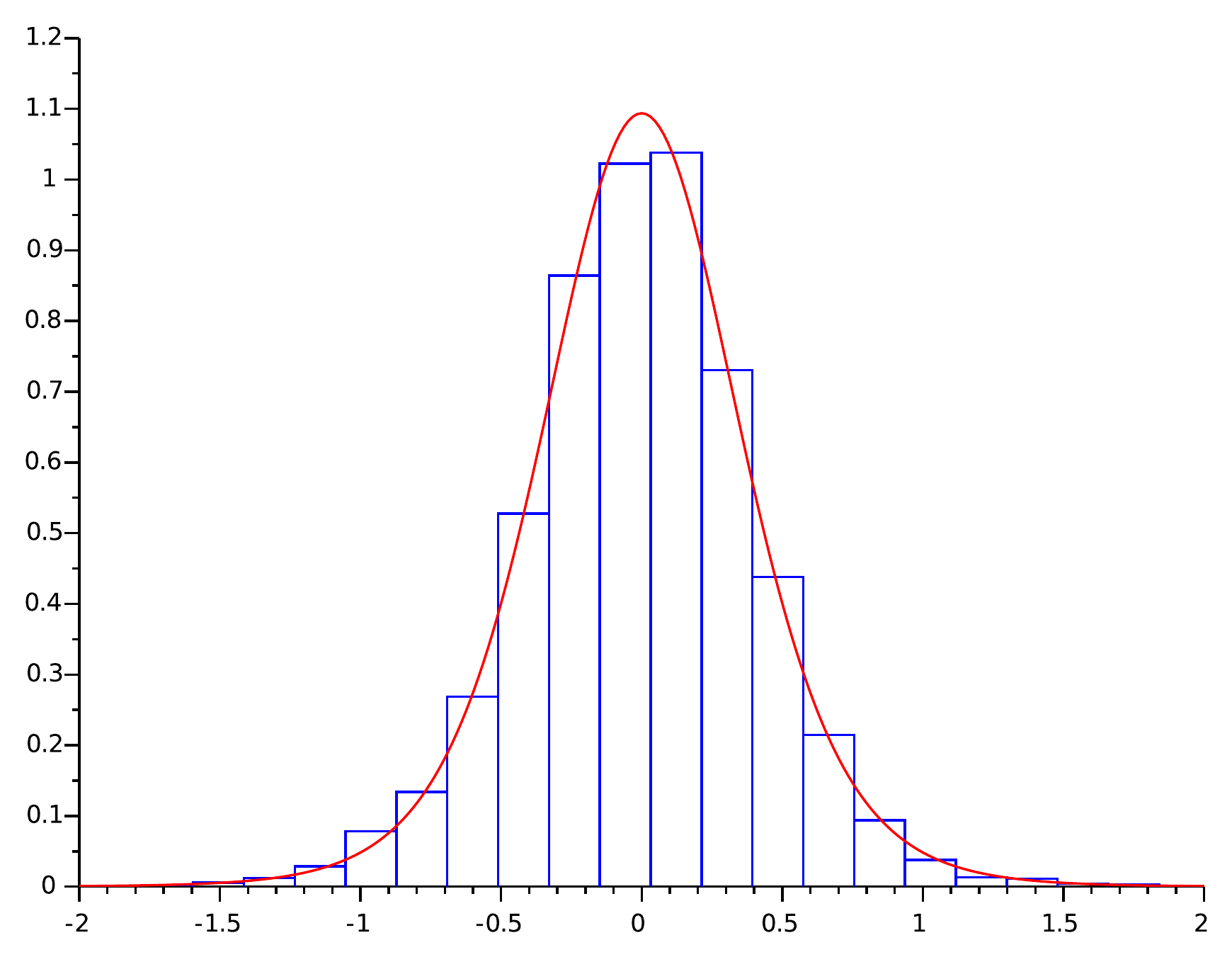}\label{figIVd}}
\caption{Histograms of Simulations of $P_1$  at Large Times}\label{figIV}
\end{figure}

Also for components $V_i$, $i=1$, $2$, $3$, of the Ornstein--Uhlenbeck velocity process 
$V$ (see~\eqref{ROUV}) it is straightforward to calculate their marginal invariant densities:

\begin{corollary} \label{cor_inv_meas_Vi}
For $d=3$ and $\gg>2$ the cartesian components $V_i$, $i=1$, $2$, $3$, of the Ornstein--Uhlenbeck 
velocity process have marginal invariant densities $\vp_{V_i}$ relative to Lebesgue measure 
given by a symmetric Beta law on $[-1,1]$ with parameter $\gg/2$. Explicitly:
\begin{equation}	\label{eq_v_dens}
	\vp_{V_i} (v) = \frac{\gG(\gg)}{2^{\gg-1} \gG(\gg/2)^2}\,(1-v^2)^{\gg/2 -1},
							\qquad i=1,2,3,\, v\in [-1,1].
\end{equation}
\end{corollary}
Figure~\ref{figV} shows histograms of $V_2$ resulting from the simulation of the momentum
process, where the red line is the graph of the Beta density~\eqref{eq_v_dens}.
\begin{figure} \centering
\subfloat[$d=3$, $\gg=4$]
	{\includegraphics[scale=.3]{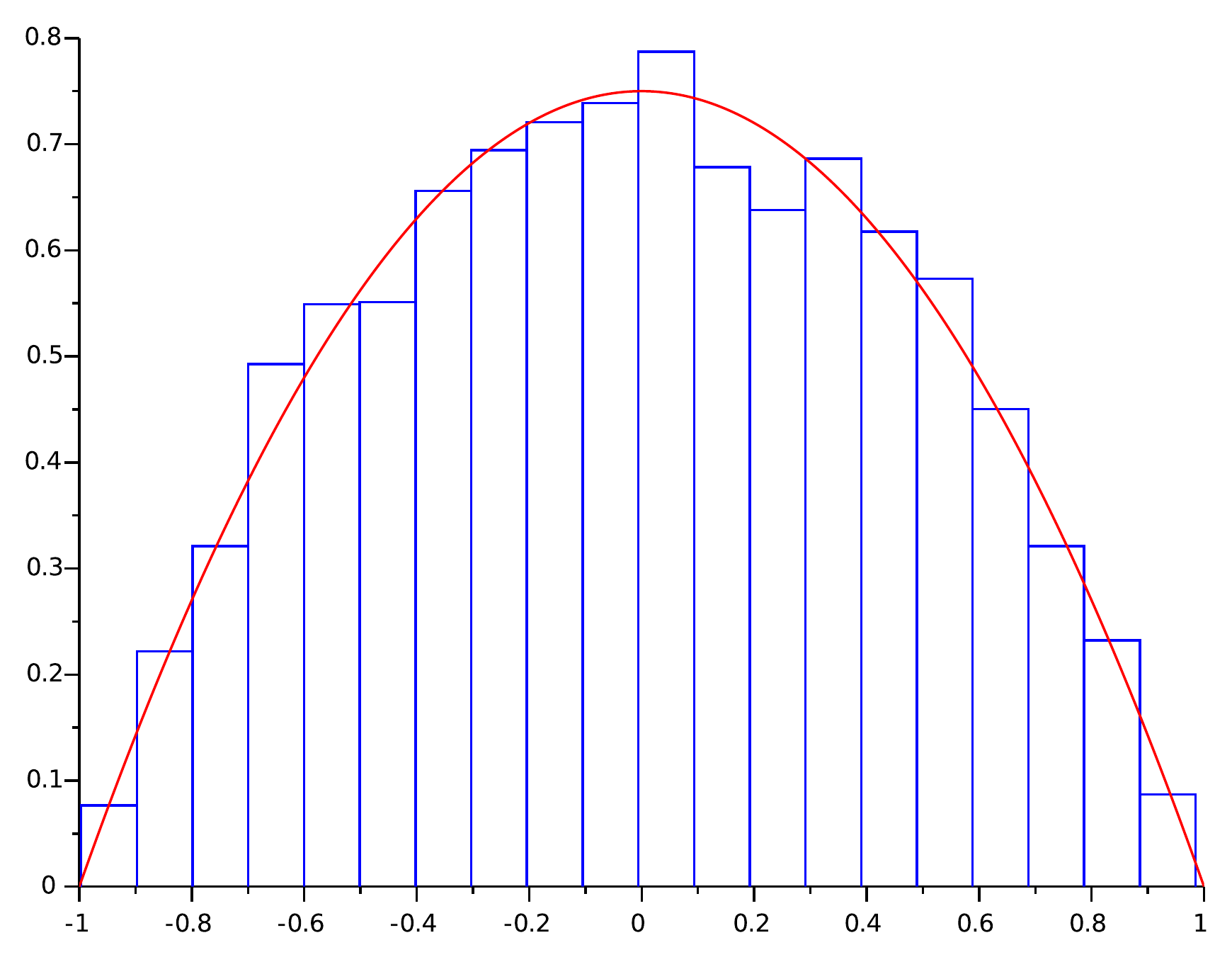}\label{figVa}}
\hspace{4em}
\subfloat[$d=3$, $\gg=6$]
	{\includegraphics[scale=.3]{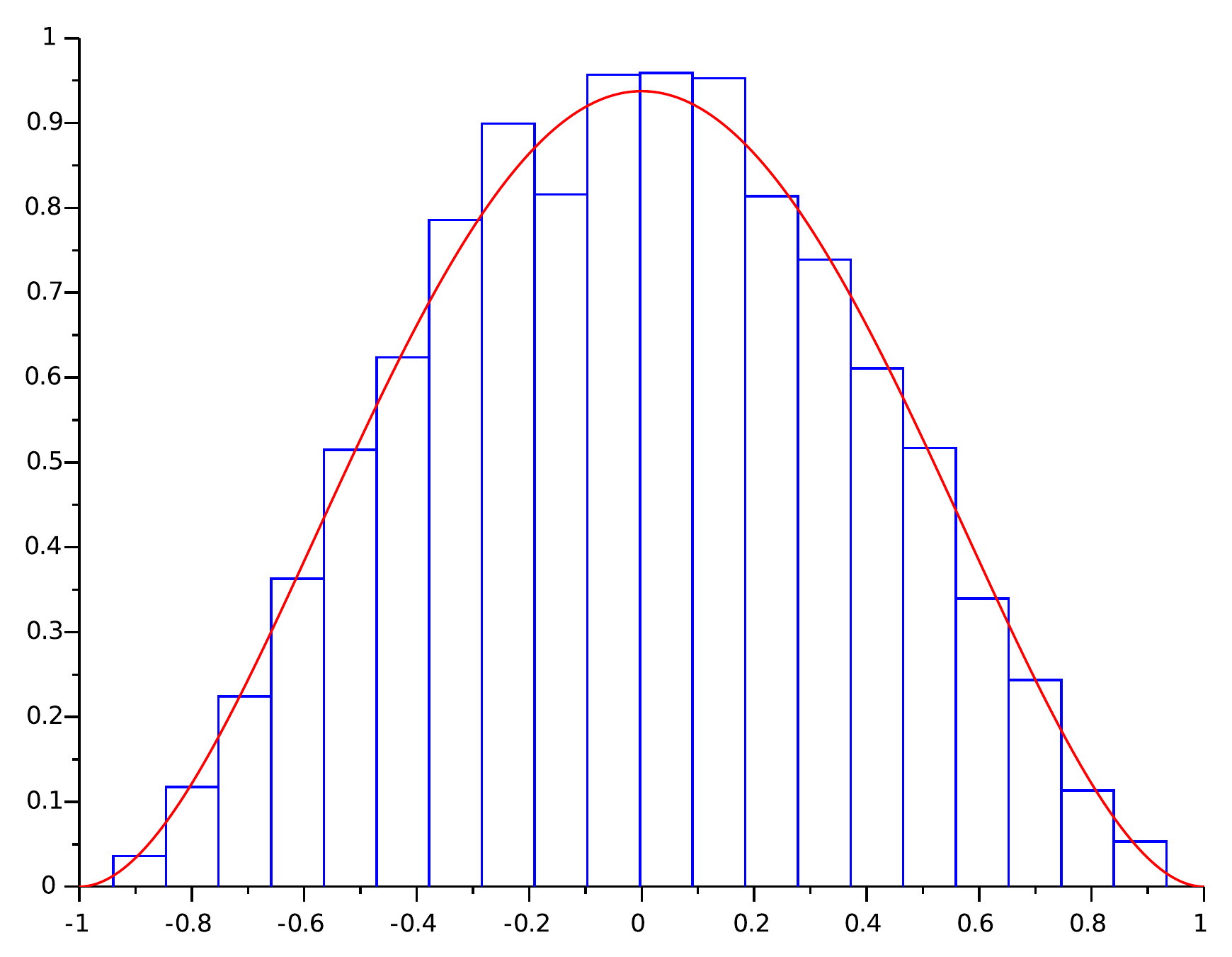}\label{figVb}}
\newline
\subfloat[$d=3$, $\gg=8$]
	{\includegraphics[scale=.3]{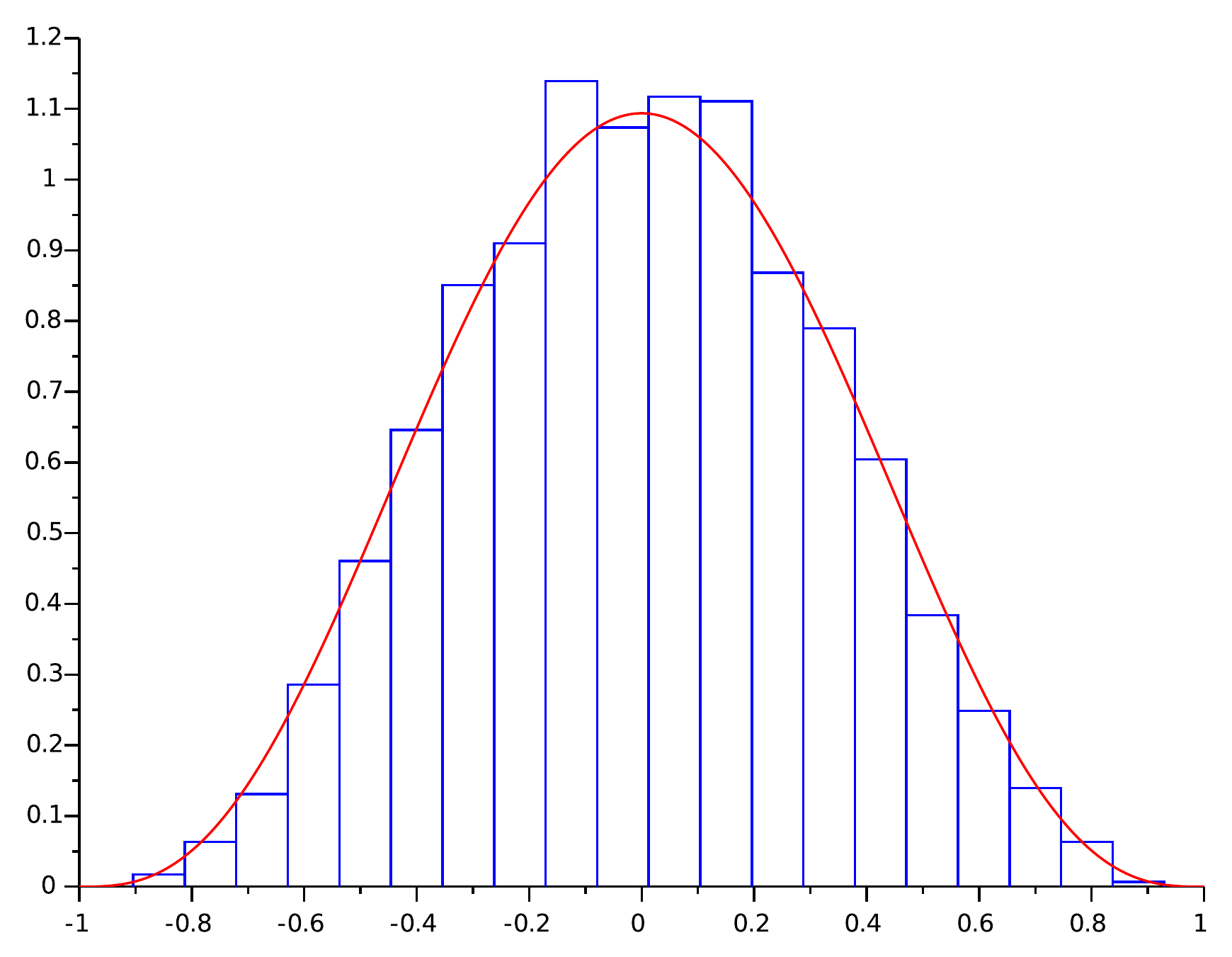}\label{figVc}}
\hspace{4em}
\subfloat[$d=3$, $\gg=10$]
	{\includegraphics[scale=.3]{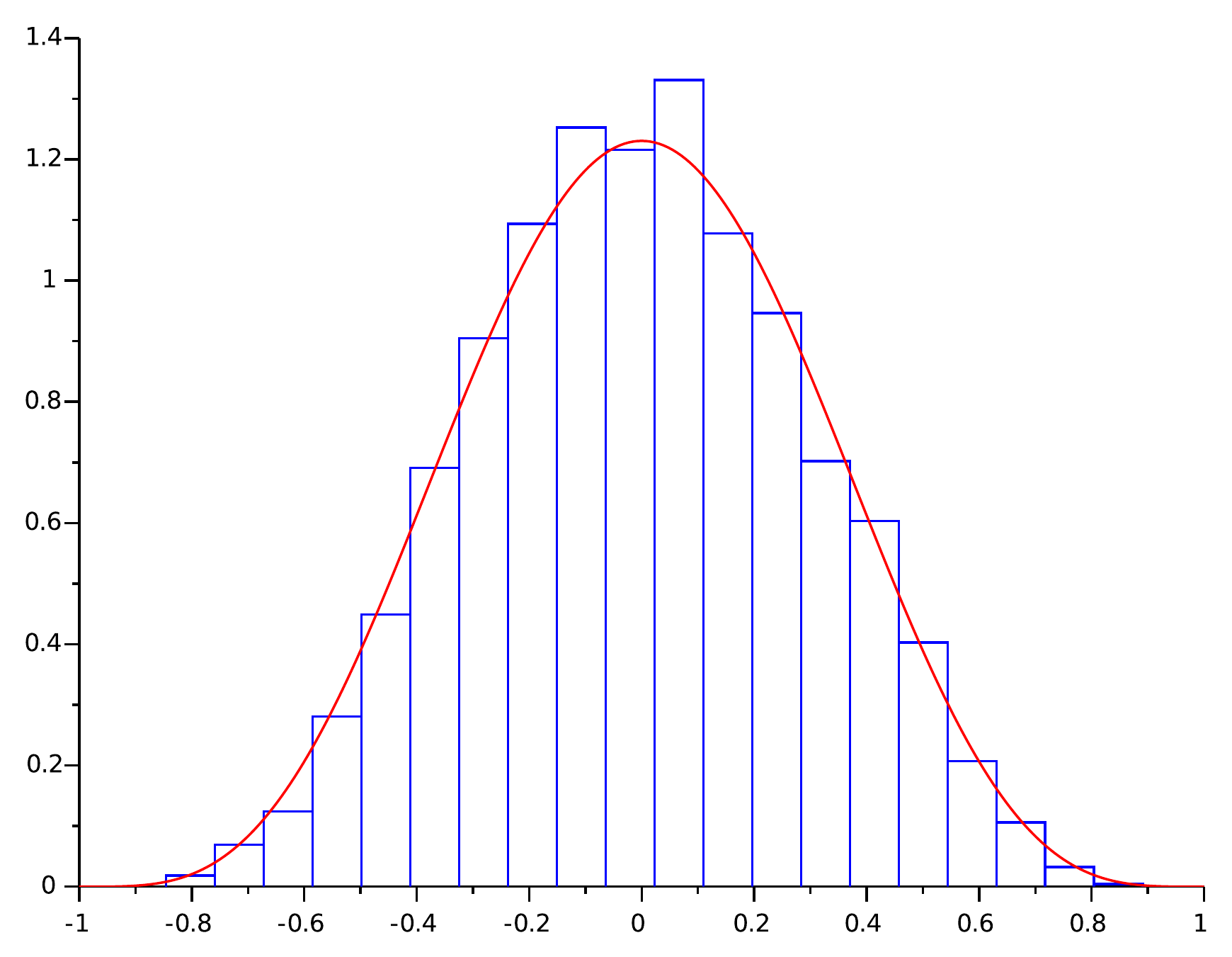}\label{figVd}}
\caption{Histograms of Simulations of $V_2$  at Large Times}\label{figV}
\end{figure}

Remark that for small velocities $v$, the density
in~\eqref{eq_v_dens} is (up to normalization) close to a centered normal
distribution:
\begin{equation*}
	(1-v^2)^{\gg/2 -1} \approx e^{-(\gg/2-1) v^2}
\end{equation*}
so that at least for the cartesian components of the velocity process we obtain
a certain compatibility with the invariant density of the classical, non-relativistic 
Ornstein--Uhlenbeck  process.

\begin{appendix}

\section{Simulations}	\label{app_Sim}
Consider the stochastic differential equations~\eqref{SDE_S} and~\eqref{SDE_tanh}.
As argued in sections~\ref{sect_BM} and~\ref{sect_OU}, the solutions do not leave
the interval $(0,+\infty)$ when started there. However, if one tries to simulate paths 
with a naive scheme, such as the Euler--Maruyama scheme, it is not possible to prevent 
all paths from crossing the singularity of the drift at $s=0$ into the region $(-\infty,0)$.
The reason is of course, that one actually simulates a random walk, and the discrete
increments do have the possibility to cross the singularity at $s=0$. A simulation 
scheme, called \emph{backward Euler--Maruyama} (BEM) scheme, which does prevent this 
crossing has been provided by Neuenkirch and Szpruch in~\cite{NeSz14}. The conditions 
formulated in~\cite{NeSz14} for their results to hold are fulfilled by the SDE's
considered in the present article. The scheme is of the form
\begin{equation}	\label{eq_BEM}
	s_{t+\gD t} = s_t + b(s_{t+\gD t})\,\gD t + \gD W_{t+\gD t},
\end{equation}
where $b$ is the drift, and the increments $\gD W_{t+\gD t}$ of the Wiener process 
are --- as usual --- independent centered normal variates with variance $\gD t$.
Observe that in order to compute an increment of $s$ from one time step to the
next, one has to numerically solve an \emph{implicit} problem. We implemented
this scheme in \textsf{Scilab},
\footnote{\texttt{http://www.scilab.org}}
and had to observe that sometimes we still obtained paths which crossed the singularity 
of the drift at the origin. A careful analysis showed that this is due to the fact that 
\textsf{Scilab}'s \texttt{fsolve} routine does not in all cases find the correct solution 
of the implicit problem. We believe that this is so because probably \textsf{Scilab}'s 
\texttt{fsolve}  is based on Newton's method, which is well-known to fail under certain 
circumstances. In order to get a functioning scheme for the SDE's~\eqref{SDE_S}, 
\eqref{SDE_tanh}, we therefore supplemented  \textsf{Scilab}'s \texttt{fsolve} with a 
bisection method for those cases, where a jump across the singularity had occured. 
Actually, a similar consideration had to be done for the SDE of $\Theta$ 
in~\eqref{SDE_d_3}, in which case the drift has singularities at $\theta=0$ 
and~$\theta=\pi$.

For each of the histograms in figures~\ref{figI} and~\ref{figII}, we generated with 
the above described method samples of $5\times 10^3$ paths, with $m^2=1$, $\gD t = 2^{-6}$, 
and let the paths develop until time $\tau=100$, i.e., altogether over $2^6\times 100$ 
time steps.

\end{appendix}

\vspace*{1\baselineskip}
\noindent
\textbf{Afterword by JP.} During the work on this manuscript my mentor and
coauthor \emph{Robert Schrader} passed away. Robert was a truly outstanding 
scientist, a charismatic teacher, and a wonderful colleague and friend --- 
I miss him very much. 

\providecommand{\bysame}{\leavevmode\hbox to3em{\hrulefill}\thinspace}
\providecommand{\MR}{\relax\ifhmode\unskip\space\fi MR }
\providecommand{\MRhref}[2]{%
  \href{http://www.ams.org/mathscinet-getitem?mr=#1}{#2}
}
\providecommand{\href}[2]{#2}


\begin{thebibliography}{10}

\bibitem{Ba10}
I.~Bailleul, \emph{{A} stochastic approach to relativistic diffusions}, Annales
  de l’Institut Henri Poincar\'e - Probabilit\'es et Statistiques \textbf{46}
  (2010), 760--795.

\bibitem{ChEn05}
A.~S. Cherny and H.-J. Engelbert, \emph{{S}ingular {S}tochastic {D}ifferential
  {E}quations}, Lecture Notes in Mathematics, no. 1858, Springer Verlag,
  Berlin, Heidelberg, New York, 2005.

\bibitem{DeMa97}
F.~Debbasch, K.~Mallick, and J.~P. Rivet, \emph{{R}elativistic
  {O}rnstein-{U}hlenbeck {P}rocess}, J. Stat. Physics \textbf{88} (1997),
  945--966.

\bibitem{Du65}
R.~M. Dudley, \emph{{L}orentz-invariant {M}arkov processes in relativistic
  phase space}, Arkiv f. Matematik \textbf{6} (1965), 241--268.

\bibitem{DuHa05}
J.~Dunkel and P.~H\"anggi, \emph{{T}heory of relativistic {B}rownian motion:
  {T}he (1 + 1)-dimensional case}, Phys. Rev. E \textbf{71} (2005), 016124.

\bibitem{DuHa05a}
\bysame, \emph{{T}heory of relativistic {B}rownian motion: {T}he (1 +
  3)-dimensional case}, Phys. Rev. E \textbf{72} (2005), 036106.

\bibitem{DuHa09}
\bysame, \emph{{R}elativistic {B}rownian {M}otion}, arxiv:0812.1996v2, 2009.

\bibitem{EeEl76}
J.~Eells and K.~D. Elworthy, \emph{{S}tochastic dynamical systems}, Control
  Theory and Topics in Functional Analysis, III (Vienna), Intern. atomic
  enegery agency, 1976, pp.~179--185.

\bibitem{El82}
K.~D. Elworthy, \emph{{S}tochastic {D}ifferential {E}quations on {M}anifolds},
  Cambridge Univ.~Press, Cambridge, 1982.

\bibitem{Fr09}
J.~Franchi, \emph{{R}elativistic diffusion in {G}\"odel’s universe}, Commun.
  Math. Phys. \textbf{290} (2009), 523--555.

\bibitem{Fr14}
\bysame, \emph{{F}rom {R}iemannian to relativistic diffusions}, Tech. report,
  IRMA, Univ. Strasbourg, 2014.

\bibitem{FrLe06}
J.~Franchi and Y.~Le~Jan, \emph{{R}elativistic diffusions and {S}chwarzschild
  geometry}, Commun. Pure Appl. Math. \textbf{LX} (2006), 187--251.

\bibitem{FrLe11}
\bysame, \emph{{C}urvature diffusions in general relativity}, Commun. Math.
  Phys. \textbf{307} (2011), 351--382.

\bibitem{FrLe12}
\bysame, \emph{{H}yperbolic dynamics and {B}rownian motion}, Oxford
  Mathematical Monographs, Oxford, 2012.

\bibitem{Ha09}
Z.~Haba, \emph{{R}elativistic diffusion}, Phys. Rev. E \textbf{79} (2009),
  021128.

\bibitem{Ha10}
\bysame, \emph{{R}elativistic diffusion with friction on a pseudo-{R}iemannian
  manifold}, Class. Quantum Grav. \textbf{27} (2010), 095021.

\bibitem{HaTh94}
W~Hackenbroch and A.~Thalmaier, \emph{{S}tochastictische {A}nalysis}, Teubner,
  Stuttgart, 1994.

\bibitem{He78}
S.~Helgason, \emph{{D}ifferential {G}eometry, {L}ie {G}roups, and {S}ymmetric
  {S}paces}, Academic Press, New York, 1978.

\bibitem{He09}
J.~Herrmann, \emph{{D}iffusion in the special theory of relativity}, Phys. Rev.
  E \textbf{80} (2009), 051110.

\bibitem{He10}
\bysame, \emph{{D}iffusion in the general theory of relativity}, Phys. Rev. D
  \textbf{82} (2010), 024026.

\bibitem{Hs02}
E.~P. Hsu, \emph{{S}tochastic {A}nalysis on {M}anifolds}, Graduate Studies in
  Math., vol.~38, American Math. Soc., Providence, 2002.

\bibitem{IkWa89}
N.~Ikeda and S.~Watanabe, \emph{{S}tochastic {D}ifferential {E}quations and
  {D}iffusion {P}rocesses}, 2nd ed., North Holland, Amsterdam, Oxford, New
  York, 1989.

\bibitem{It50}
K.~It\^o, \emph{{S}tochastic differential equartions in a differentiable
  manifold}, Nagoya Math. J. \textbf{1} (1950), 35--47.

\bibitem{ItMc74}
K.~It\^o and H.~P. McKean~Jr., \emph{{D}iffusion {P}rocesses and their {S}ample
  {P}aths}, 2nd ed., Springer, Berlin, Heidelberg, New York, 1974.

\bibitem{Jo11}
J.~Jost, \emph{{R}iemannian {G}eometry and {G}eometric {A}nalysis}, 6th ed.,
  Springer Verlag, Berlin, Heidelberg, New York, 2011.

\bibitem{KaSh91}
I.~Karatzas and S.~E. Shreve, \emph{{B}rownian {M}otion and {S}tochastic
  {C}alculus}, 2nd ed., Springer, Berlin, Heidelberg, New York, 1991.

\bibitem{Ma78}
P.~Malliavin, \emph{{G}\'eometrie {D}iff\'erentielle {S}tochastique}, Presse de
  l' Universit\'e de Montr\'eal, Montr\'eal, 1978.

\bibitem{Ne67}
E.~Nelson, \emph{{D}ynamical {T}heories of {B}rownian {M}otion}, Princeton
  Univ. Press, 1967.

\bibitem{NeSz14}
A.~Neuenkirch and L.~Szpruch, \emph{{F}irst order strong approximation of
  scalar {SDE}s defined in a domain}, Numer. Math. \textbf{128} (2014),
  103--136.

\bibitem{ReYo91}
D.~Revuz and M.~Yor, \emph{{C}ontinuous {M}artingales and {B}rownian {M}otion},
  Springer, Berlin, Heidelberg, New York, 1999.

\bibitem{St71}
D.~W. Stroock, \emph{{O}n the growth of stochastic {I}ntegrals}, Z.
  Wahrscheinlichkeitstheorie verw. Geb. \textbf{18} (1971), 340--344.

\bibitem{UhOr30}
G.~E. Uhlenbeck and L.~S. Ornstein, \emph{{O}n the theory of {B}rownian
  motion}, Phys. Rev. \textbf{36} (1930), 823--841.

\bibitem{YaWa71}
T.~Yamada and S.~Watanabe, \emph{{O}n the uniqueness of solutions of stochastic
  differential equations}, J. Math. Kyoto Univ. \textbf{11} (1971), 155--167.

\bibitem{Yo49}
K.~Yosida, \emph{{B}rownian motion on the surface of the 3-sphere}, Ann. Math.
  Statistics \textbf{20} (1949), 292--296.

\bibitem{Yo52}
\bysame, \emph{{B}rownian motion in a homogeneous {R}iemannian space}, Pac
  \textbf{2} (1952), 263--270.

\end{thebibliography}
\end{document}